\documentclass[smallextended]{svjour3}

\usepackage[paperwidth=210mm,paperheight=297mm,centering,hmargin=2.3cm,vmargin=2.7cm]{geometry}
\usepackage{marvosym}
\usepackage{amsfonts}
\usepackage{dcolumn}
\usepackage{bm,bbm}
\usepackage{kpfonts}
\usepackage{braket}
\usepackage{color}
\usepackage{authblk}

\newcommand{\cC}{{\mathcal C}}
\newcommand{\cD}{{\mathcal D}}
\newcommand{\cE}{{\mathcal E}}

\newcommand{\cH}{{\mathcal H}}

\newcommand{\cK}{{\mathcal K}}
\newcommand{\cL}{{\mathcal L}}
\newcommand{\cM}{{\mathcal M}}
\newcommand{\cN}{{\mathcal N}}

\newcommand{\cP}{{\mathcal P}}

\newcommand{\cS}{{\mathcal S}}
\newcommand{\cT}{{\mathcal T}}

\newcommand{\cX}{{\mathcal X}}
\newcommand{\cY}{{\mathcal Y}}

\newcommand{\fE}{{\mathfrak E}}
\newcommand{\fF}{{\mathfrak F}}

\newcommand{\fI}{{\mathfrak I}}

\newcommand{\fP}{{\mathfrak P}}

\newcommand{\fS}{{\mathfrak S}}
\newcommand{\fT}{{\mathfrak T}}

\newcommand{\fW}{{\mathfrak W}}

\newcommand{\bbmC}{{\mathbbm C}}
\newcommand{\bbmE}{{\mathbbm E}}
\newcommand{\bbmN}{{\mathbbm N}}

\newcommand{\bbmeins}{{\mathbbm 1}}

\newcommand{\prob}{\mathrm{Pr}}
\newcommand{\id}{\mathrm{id}}

\newcommand{\tr}{\mathrm{tr}}
\newcommand{\spann}{\mathrm{span}}

\newcommand{\conv}{\mathrm{conv}}
\newcommand{\pr}{\mathfrak{P}}

\newcommand{\rebd}{\mathrm{rebd}}
\newcommand{\ri}{\mathrm{ri}}
\newcommand{\aff}{\mathrm{aff}}

\begin{document}
\title{Entanglement-assisted classical capacities of compound and arbitrarily varying quantum channels} 
\author{Holger Boche \and Gisbert Jan\ss en \and Stephan Kaltenstadler\thanks{The authors thank their colleague Sajad Saeedinaeeni for careful reading of the 
manuscript and suggestions which improved the presentation of this paper. This work was supported by Bundesministerium f\"ur Bildung und Forschung (BMBF) via grant 16KIS0118K (H.B., G.J.) and the German National Science Foundation (DFG) via project
Bo 1743/ 20-1 (S.K.).}}
\institute{Holger Boche \at  \email{boche@tum.de} \and Gisbert Jan\ss en \at \email{gisbert.janssen@tum.de} \and Stephan Kaltenstadler \at \email{kaltenstadler@campus.tu-berlin.de} \at \\
Lehrstuhl f\"ur Theoretische Informationstechnik, Technische Universit\"at M\"unchen, 80290 M\"unchen, Germany}
\titlerunning{Entanglement-assisted classical capacities of compound and AV quantum channels}
\authorrunning{H. Boche, G. Jan\ss en, S. Kaltenstadler}
\maketitle

\begin{abstract}
We consider classical message transmission under entanglement assistance for compound memoryless and arbitrarily varying quantum channels. In both cases, we prove general coding theorems together with corresponding 
weak converse bounds. In this way, we obtain single-letter characterizations of the entanglement-assisted classical capacities for both channel models. Moreover, we show that the 
entanglement-assisted classical capacity does exhibit no strong converse property for some compound quantum channels for the average as well as the maximal error criterion. 
A strong converse to the entangle- ment-assisted classical capacities does hold for each arbitrarily varying quantum channel.
\end{abstract}
\begin{section}{Introduction} \label{sect:introduction}
Entanglement is well-known as a valuable communication resource in quantum information theory. Beside several tasks, such as quantum teleportation \cite{bennett93}, where entanglement is an indispensable resource 
to run successful protocols, entanglement also has impact as an additional resource in quantum channel coding scenarios.
An early example where additional entanglement has such an effect is the noiseless dense-coding protocol \cite{bennett92}, where a shared
maximally entangled pair of qubits in addition to a noiseless qubit channel allow perfect transmission of four bits. That is two times as much as when doing message transmission over that channel without further entanglement. 
A Shannon-theoretic refinement of this idea was introduced in \cite{bennett02} (see also \cite{holevo02}), where the classical message transmission capacity of (potentially noisy) memoryless quantum channels 
under free supply of additional shared entanglement 
was determined. The entanglement-assisted classical message transmission capacity of a memoryless quantum channel was characterized by the input-state maximized quantum mutual information of that channel. 
This result has a blood-stirring effect on the information-theorist for at least two reasons. \newline 
On one hand, additional entanglement allows to achieve substantially higher transmission rates for some channels. On the other hand, the capacity is characterized by a handy single-letter formula, 
a feature which is not shared by most of the known capacities for quantum communication tasks. 
\newline 
An extension of the mentioned results to multi-user situations led the authors of \cite{hsieh08} to a characterization of the entanglement-assisted message transmission capacity region of memoryless quantum multiple-access 
channels.
The codes derived therein to prove the coding theorem shed further light on the utility of entanglement-assisted 
message transmission codes. It is possible to derive coding constructions and protocols for several other important quantum communication tasks, by making entanglement-assisted classical message transmission codes coherent 
\cite{devetak08}, \cite{hsieh10}. 
In this way, entanglement-assisted message transmission codes fill a prominent position within the so-called ``family of quantum protocols'' \cite{devetak04}. \newline 
All mentioned results were derived under the idealized conditions that the transmission channel is memoryless, and the generating channel map which governs the statistics of the system is perfectly known to sender and 
receiver. Both of the mentioned restrictions will be hardly fulfilled in real-world communication systems. In this paper, we pursue a way in direction of partly dropping the mentioned conditions. We investigate 
the task of entanglement-assisted message transmission assuming the users to be connected by either a compound memoryless quantum channel \footnote{During the final stage of preparation of this work, the authors became aware 
of the paper \cite{berta16}. Some results we present regarding the entanglement-assisted classical capacity of compound quantum channels are 
also contained therein, with substantially different proof techniques.} or an arbitrarily varying quantum channel (AVQC).
\newline 
If the communication parties are confronted with a compound memoryless quantum channel, the transmission is governed by memoryless extensions of a generating channel map for each blocklength. However, sender and receiver
have no perfect knowledge of the actual generating map. They are rather provided with a set of confidence of channel maps, where each of them is possibly generating the transmission. 
Therefore, they are forced to use coding procedures that are 
universal in the sense that they are asymptotically perfectly reliable for each of the possible realizations of the channel statistics. \newline 
The AVQC model confronts the users with a substantially increased level of system uncertainty. Each use of the channel can be driven by an arbitrary channel map from a prescribed 
set of channels, where most of the possible realizations are not even memoryless. It is instructive for the AVQC model to regard a third, malicious party being involved in the scenario. This third party acting 
as a jammer may choose the channel map for each use of the channel freely from a prescribed set to attack the transmission goals of the sending and receiving parties connected by the channel. \newline 
The contributions of the paper are the following. We prove a coding theorem for entanglement-assisted message transmission over any given compound memoryless quantum channel, and a corresponding converse bound
which determine the entanglement-assisted classical capacity of compound memoryless channels for the average
as well as the maximal transmission error as a criterion in terms of a single-letter formula.  \newline 
Considering the AVQC model, we use the entanglement-assisted message transmission codes derived for the compound quantum channel. Applying 
quantum versions of the so-called robustification and elimination methods from \cite{ahlswede86}, we prove a coding theorem for entanglement-assisted message transmission over AVQCs. Together with the corresponding 
converse bound, we establish a single-letter characterization of the entanglement-assisted message transmission capacity also for the AVQC. \newline 
From the obtained capacity characterization, we infer two remarkable features of the entanglement-assisted classical capacities of AVQCs. 
The capacity is additive, and continuous, which implies stability of the capacity under perturbation of the AVQC-generating set of channels. Both do not hold in general for the unassisted classical message transmission 
capacities of AVQCs \cite{boche13}. \newline 
The question whether or not several coding theorems in quantum information theory can be supplied with a so-called strong converse recently has received increasing interest among quantum information theorists. We show, by providing 
a counterexample that if the average transmission error is considered as criterion of reliability, no such strong converse can hold in general for the entanglement-assisted classical capacity of compound channels. A general 
strong converse statement does also not hold for the maximal error-criterion, which we show by demonstrating that under entanglement assistance both error criteria are essentially equivalent for compound quantum channels. 
We complete the set of statements on both channel models by providing a general strong converse statement for AVQCs. This is one more remarkable feature of entanglement assistance. For the unassisted classical capacities of 
AVQCs it is an open question, whether or not such a statement does hold. Even for classical AVQCs a general strong converse for the unassisted message transmission capacity, by now, is not more than a conjecture \cite{ahlswede06}.

\begin{subsection}*{Related work} \label{subsect:related_work}
The task of entanglement-assisted message transmission was first considered by Bennett et al. \cite{bennett02} (see also \cite{holevo02}). Therein, the classical message transmission capacity of a perfectly known
memoryless quantum channel was determined. The multi-user capacity for entanglement-assisted message transmission was characterized in \cite{hsieh08}, where a refined coding strategy
was presented. The coding theorem for (unassisted) classical message transmission over compound
quantum channels was derived in \cite{mosonyi15}, while a coding theorem for the genuine quantum capacities without entanglement assistance of a compound quantum channel were proven in 
\cite{bjelakovic09b}. Later on, the same authors together with R. Ahlswede also derived coding theorems for the unassisted quantum capacities of arbitrarily varying quantum channels \cite{ahlswede13}. The techniques 
used in this paper, in fact, strongly rely on the arguments used therein. To prove a coding theorem for the entanglement-assisted classical capacities of compound channels, we use capacity achieving codes for certain 
compound classical-quantum channels. Such codes were derived in \cite{hayashi09}, \cite{bjelakovic09}, \cite{datta10} before. In this paper, we use codes from the more recent work \cite{mosonyi15} instead,
which achieve the message transmission capacities of general (not necessarily finite or countable compound channels) with exponentially decreasing errors.
From \cite{ahlswede13} we borrow a variation of the famous robustification and elimination techniques which in turn originally were introduced in \cite{ahlswede86},
as a method to prove the coding theorem for arbitrarily varying classical channels. Another, very nice approach to derive good codes for entanglement-assisted classical message transmission can be found in \cite{hsieh08},
where the capacity region for entanglement-assisted message transmission over quantum multiple access channels was derived. The coding strategy to prove the latter result relied on a packing lemma together with a very powerful while elementary encoder construction,
which also added a nice method of proof for the single user setting. We exploit this approach, and show that the encoding construction is also reasonable to derive sufficient codes for entanglement-assisted message transmission 
over compound quantum channels. \newline 
It was shown in \cite{bennett14} that the entanglement-assisted classical capacity even obeys a general strong converse property for perfectly known memoryless quantum channels, i.e. all code sequences achieving rates above 
capacity are asymptotically completely useless (i.e. their transmission errors approach one in the asymptotic limit). A different proof for this result was given recently in \cite{gupta14}. \newline 
While the present paper was assembled, the authors learned of the paper \cite{berta16}, which has some overlap in results with the present one. Therein, several entanglement-assisted capacities of compound quantum channels were 
determined. However, the results were proven there with a different approach, employing results from one-shot information theory to derive a coding theorem for entanglement-assisted entanglement transmission, which in turn
led them to a proof of the classical entanglement-assisted capacities. Here, we take the opposite route. We derive a coding theorem for entanglement-assisted message transmission from universal codes for unassisted message 
transmission over compound classical-quantum channels. Moreover, the techniques employed in this work allow to derive codes for entanglement-assisted message transmission over compound quantum channels which have sufficient 
error performance to use a variant of the robustification approach \cite{ahlswede78} which allows us to prove a coding theorem for arbitrarily varying quantum channels. To derive coding theorems regarding the entanglement-assisted 
capacities for the latter channel model is also stated as open problem in \cite{berta16}. 
\end{subsection}
\end{section}
\begin{section}{Notation and conventions}\label{sect:notations}
 All Hilbert spaces appearing in this work are considered to be finite dimensional complex vector spaces. 
 $\mathcal{L}(\cH)$ is the set of linear maps and $\cS(\cH)$ the set of states (density matrices) on a Hilbert 
 space $\cH$ in our notation. We denote the set of quantum channels, i.e. completely positive and trace 
 preserving (c.p.t.p.) maps from $\mathcal{L}(\cH)$ to $\mathcal{L}(\cK)$ by $\mathcal{C}(\cH, \cK)$.
 \newline
 Regarding states on multiparty systems, we freely make use of the following convention for a system consisting
 of some parties $X,Y,Z$, for instance, we denote $\cH_{XYZ} := \cH_{X} \otimes \cH_Y \otimes \cH_Z$, and denote
 the marginals by the letters assigned to subsystems, i.e. $\sigma_{XZ} := \tr_{\cH_Y}(\sigma)$ for $\sigma \in
 \cS(\cH_{XYZ})$ and so on. \newline 
 The von Neumann entropy of a quantum state $\rho$ is 
 defined 
 \begin{align}
  S(\rho) := - \tr(\rho \log \rho),
 \end{align}
 where we denote by $\log(\cdot)$ and $\exp(\cdot)$ the base two logarithms and exponentials throughout this paper.
 The \emph{quantum mutual information} of a quantum state $\rho \in \cS(\cH_A)$, and a channel $\cN \in \cC(\cH_A, \cH_B)$ is 
 defined by
 \begin{align}
  I(\rho, \cN) := S(\rho) + S(\cN(\rho)) - S(\cN \otimes \id(\ket{\psi}\bra{\psi})), \label{quant_mut_inf_def}
 \end{align}
 where $\psi$ is the state vector of an arbitrary purification of $\rho$. The quantum mutual information is well-defined by 
 (\ref{quant_mut_inf_def}), because the r.h.s. is known to not depend on the  choice of the purification $\psi$. 
  We denote the set of classical probability distributions on a set $S$ by $\pr(S)$. The $l$-fold Cartesian
 product of $S$ will be denoted $S^l$ and $s^l := (s_1,...,s_l)$ will be a notation for elements of 
 $S^l$. For each positive integer $n$, the shortcut $[n]$ is used to abbreviate the set $\{1,...,n\}$.
 For a set $A$ we denote the convex hull of $A$ by $\conv(A)$. If $\fI := \{\cN_s\}_{s \in S} \subset \cC(\cH,\cK)$ is 
 a finite set of quantum channels, the convex hull can be written as 
 \begin{align}
  \conv(\mathcal{\fI}) = \left\{\tilde{\cN}_p \in \cS(\cH):\ \tilde{\cN}_p = \sum_{s\in S} p(s) \ \cN_s, \ p \in \pr(S)
    \right\}. \label{conv_hull_def}
 \end{align}
 We use the diamond norm $\|\cdot\|_\Diamond$ to measure the distance between quantum channels. For a linear map 
 $\cN: \cL(\cH) \rightarrow \cL(\cK)$, its diamond norm is defined 
 \begin{align}
  \|\cN\|_\Diamond := \underset{n \in \bbmN}{\sup} \ \underset{\substack{a \in \cL(\bbmC^n \otimes \cH) \\ \|a\|_1 = 1}}{\max} \| \id_{\bbmC^n} \otimes \cN (a)\|_1,
 \end{align}
where $\id_{\bbmC^n}$ is our notation for the identical channel, i.e. $\id_{\bbmC^n}(x) = x$ for each $x \in \cL(\bbmC^n)$. By $D_\Diamond$, we denote the 
Hausdorff distance which is generated by $\|\cdot\|_\Diamond$. For any two sets $\fI, \fI' \subset \cC(\cH, \cK)$, their Hausdorff distance is defined by
\begin{align}
 D_\Diamond(\fI, \fI') := \max\left\{\underset{\cN \in \fI}{\sup} \underset{\cN' \in \fI'}{\inf} \|\cN - \cN' \|_\Diamond, \ \underset{\cN' \in \fI'}{\sup} \underset{\cN \in \fI}{\inf} \|\cN - \cN' \|_\Diamond \right\}.
\end{align}
 By $\fS_n$, we denote the group of permutations on $n$ elements, in this way
 $\sigma(s^n) = (s_{\sigma(1)},...,s_{\sigma(n)})$ for each $s^n = (s_1,...,s_n)\in S^n$ and permutation $\sigma \in
 \fS_n$.\newline
\end{section}

\begin{section}{Basic definitions and main results} \label{sect:definitions}
 In this section, we give concise definitions for the coding scenarios we consider, and state the main results of this paper. First we introduce 
 the compound memoryless quantum channel and AVQC models. Let for the rest of this section 
 $\fI:= \{\cN_s\}_{s \in S} \subset \cC(\cH_A, \cH_B)$ be a given set of c.p.t.p. maps with a set $S$ of parameters not necessarily finite or countable. \newline 
 The \emph{compound quantum channel generated by }$\fI$ is given by the set $\{\cN_s^{\otimes n}: s \in S, n \in \bbmN\}$. This definition 
 is understood as follows. For each blocklength $n$, the transmission is governed by $\cN_s^{\otimes n}$, where $s$ can be any member of the index set $S$. \newline 
 The \emph{AVQC generated by} $\fI$ is given by the set $\{\cN_{s^n}: s^n \in S^n, \ n \in \bbmN\}$,
 where the definitions 
 \begin{align*}
  \cN_{s^n} := \cN_{s_1} \otimes \cdots \otimes \cN_{s_n}  && (s^n = (s_1,\dots, s_n) \in S^n)
 \end{align*}
 apply. The AVQC models a rather pessimistic transmission situation. The channel map governing in the transmission can vary over the set $\fI$ in each use of the channel. AVQCs can be thought as
 modelling an adversarial attack on the transmission. A jammer may confront the communication parties with an arbitrary channel from $\fI$ for each channel use.
\begin{subsection}{Compound memoryless quantum channels} \label{subsect:definitions_comp}
 In the following, we define the coding procedures allowed for entanglement-assisted message transmission. 
 \begin{definition} \label{def:comp_ea_codes}
   An \emph{$(n,L,M)$-code for entanglement-assisted (EA) message transmission} over the compound quantum channel $\fI$ is a triple $\cC := (\Psi,\cE_m,D_m)_{m =1}^M$, where with additional Hilbert 
   spaces $\cK_A$, $\cK_B$ (under control of $A$,$B$)
   \begin{itemize}
    \item $\Psi \in \cS(\cK_A \otimes \cK_B)$ \ is a pure quantum state, 
    \item $\cE_m \in \cC(\cK_A, \cH_A^{\otimes n})$ \ for all $m \in [M]$, 
    \item $D_m \in \cL(\cH_B^{\otimes n}\otimes \cK_B)$, \ $0 \leq D_m \leq \bbmeins$\ for all \ $m \in [M]$, \ $\sum_{m=1}^M D_m \leq \bbmeins$, \ and 
    \item $L := \dim \cK_A$.
   \end{itemize}
   Using the shortcut $D_m^c := \bbmeins - D_m$ for each $m\in [M]$, we define the functions
   \begin{align}
    \overline{e}(n,\cC, \fI) 	&:= \underset{s \in S}{\sup} \frac{1}{M} \sum_{m=1}^M \ \tr\left\{D^c_m(\cN_s^{\otimes n}\circ \cE_m \otimes \id_{\cK_B})(\Psi)\right\} \hspace{2cm} \text{(average error)}, \ \text{and} \nonumber \\
    e(n,\cC, \fI) 		&:= \underset{s \in S}{\sup} \underset{m \in [M]}{\max} \ \tr\left\{D^c_m(\cN_s^{\otimes n} \circ \cE_m \otimes \id_{\cK_B})(\Psi)\right\} \hspace{2cm} \text{(maximal error)}. \label{comp_max_err_def}
   \end{align}
  \end{definition}
  It might seem somewhat unusual that we regard the pure state allowed for assistance of the transmission as a feature of the code. This has its only reason in notational convenience. 
  \begin{definition} \label{def:comp_ea_capacity}
   A number $R \geq 0$ is called an \emph{achievable rate for EA message transmission over the compound quantum channel $\fI$ under average error criterion}, if for each $\epsilon > 0, \delta > 0$ there exist numbers
   $R_e < \infty$, and $n_0 = n_0(\epsilon,\delta)$, such that for each blocklength $n > n_0$ we find an $(n,L,M)$-code $\cC$ for EA message transmission over $\fI$ which has the properties
   \begin{enumerate}
    \item $\frac{1}{n} \log M \geq R - \delta$,
    \item $\overline{e}(n, \cC, \fI) \leq \epsilon$, and  
    \item $\frac{1}{n}\log L \leq R_e.$
   \end{enumerate}
   We call 
   \begin{align*}
    \overline{C}_{EA}(\fI) := \sup\{R \geq 0: \ R \ \text{achiev. rate for av. error EA message transmission over the compound channel}\ \fI \}
   \end{align*}
   the \emph{average error entanglement-assisted message transmission capacity of the compound quantum channel $\fI$}. 
  \end{definition}
  The corresponding definitions for achievable rates and capacity regarding EA message transmission over $\fI$ under maximal error criterion can be easily obtained by replacing the average error by the maximal error from 
  (\ref{comp_max_err_def}) in Definition \ref{def:comp_ea_capacity}. The corresponding capacity will be denoted by $C_{EA}(\fI)$. \newline 
  Notice that the third condition in Definition \ref{def:comp_ea_capacity} above states that only protocols are allowed, which consume entanglement on systems with rate-bounded number of degrees of freedom. 
   The upper bound on the capacity in Theorem \ref{theorem:comp_ea_capacity} below holds, in fact, also without this restriction. However,
   consuming resources of asymptotically unbounded rate seems not reasonable from the communication point of view. 
   The following theorem is the main result of this paper regarding the compound quantum channel model. 
   \begin{theorem} \label{theorem:comp_ea_capacity}
    It holds
    \begin{align}
     \overline{C}_{EA}(\fI) = C_{EA}(\fI) \ = \ \underset{\rho \in \cS(\cH_A)}{\sup} \ \underset{s \in S}{\inf} \ I(\rho, \cN_s). \label{theorem:comp_ea_capacity_1}
    \end{align}
   \end{theorem}
   The inequality $\overline{C}_{EA}(\fI) \geq C_{EA}(\fI)$ in (\ref{theorem:comp_ea_capacity_1}) follows directly from the definition of the capacities. 
   The remaining inequalities will be shown to hold in Section \ref{subsect:prfs_compound}
   below.
    We point out that a so-called strong converse to EA message transmission capacity of compound quantum channels does not hold in general for each of the error criteria.  
   To formalize this statement, we define for each $n \in \bbmN$, $\lambda \in (0,1)$, $R_e < \infty$ 
   \begin{align}
    \overline{N}_{EA}(n,\fI,R_e,\lambda) 
    &:= \max\{M: \ \exists (n,L,M)-\text{EA code}\  \cC \ \text{for} \ \fI \ \text{such that} \ \overline{e}(n,\cC, \fI) \leq \lambda, \ L \leq 2^{nR_e}\}. \label{def:comp_max_message_card}
   \end{align}
   We define $N_{EA}$ analogously by replacing the average error with the maximal error function. A strong converse holds to the average error classical message transmission capacity of the compound quantum channel $\fI$ if the following 
   statement is true. For each $\fI$, $R_e < \infty$, it holds
    \begin{align}
     \forall \lambda \in (0,1): \hspace{.3cm} \limsup_{n \rightarrow \infty} \frac{1}{n}\log \overline{N}_{EA}(n, \fI, R_e,\lambda) \ \leq  
     \underset{\rho \in \cS(\cH_A)}{\sup} \underset{\cN \in \fI}{\inf} \ I(\rho, \cN). \label{average_error_compound_strong_converse}
    \end{align}
    The above statement says that the state-maximized worst-case channel mutual information is the best achievable rate, even, if the coding 
    procedures are not demanded to approach zero average transmission error asymptotically. We show the following claim. 
    \begin{claim}
     A strong converse to $\overline{C}_{EA}(\fI)$ or $C_{EA}(\fI)$ does not hold in general.
    \end{claim}
    To justify the above claim, we demonstrate that not general strong converse does hold to $\overline{C}_{EA}(\fI)$ by giving an explicit counterexample (Example \ref{nostrongconverse_example} in Section \ref{subsect:prfs_compound}). From this assertion,
    we conclude that also in case of the maximal error, no such statement is valid in general. Indeed, Lemma \ref{lemma:comp_av_max_equal} in Section \ref{subsect:prfs_compound} states that there is no essential difference between 
    the maximal and average error criteria due to free-rate entanglement assistance. 
    \end{subsection}

  \begin{subsection}{Arbitrarily varying channels} \label{subsect:definitions_avqc}
   In this subsection, we consider an AVQC generated by the set $\fI := \{\cN_s\}_{s \in S}$. 
   \begin{definition} \label{def:avqc_ea_codes}
   An \emph{$(n,L,M)$-code for (EA) message transmission} over the AVQC $\fI$ is a triple $\cC := (\Psi,\cE_m,D_m)_{m =1}^M$, where with additional Hilbert spaces 
   $\cK_A$, $\cK_B$ (under control of $A$,$B$)
   \begin{itemize}
    \item $\Psi \in \cS(\cK_A \otimes \cK_B)$ is a pure state,
    \item $\cE_m \in \cC(\cK_A, \cH_A^{\otimes n})$ \ for all \  $m \in [M]$, 
    \item $D_m \in \cL(\cH_B^{\otimes n} \otimes \cK_B)$, \ $0 \leq D_m \leq \bbmeins$ \ for all \ $m \in [M]$, with $\sum_{m=1}^M D_m \leq \bbmeins$, \ and
    \item $L := \dim \cK_A$.
   \end{itemize}
   Using the shortcut $D_m^c := \bbmeins - D_m$ for each $m\in [M]$, we define the functions
   \begin{align}
    \overline{e}_{av}(n,\cC, \fI) 	&:= \underset{s^n \in S^n}{\sup} \frac{1}{M} \sum_{m=1}^M \ \tr\left\{D^c_m(\cN_{s^n}\circ \cE_m \otimes \id_{\cK_B})(\Psi)\right\} \hspace{2cm} \text{(average error)}, \ \text{and} \\
    e_{av}(n,\cC, \fI) 		&:= \underset{s^n \in S^n}{\sup} \underset{m \in [M]}{\max} \ \tr\left\{D^c_m(\cN_{s^n}\circ \cE_m \otimes \id_{\cK_B})(\Psi)\right\} \hspace{2cm} \text{(maximal error)}.
   \end{align}
  \end{definition}
   \begin{definition} \label{def:avqc_ea_capacity}
   A number $R \geq 0$ is called an \emph{achievable rate for EA message transmission over the AVQC $\fI$ under average error criterion}, if we find a number $R_e < \infty$, such that for each $\epsilon > 0, \delta > 0$ there 
   exists a number $n_0 = n_0(\epsilon,\delta)$, 
   such that for each blocklength $n > n_0$ we find an $(n,L,M)$-code $\cC$ for EA message transmission over $\fI$ which has the following properties
   \begin{enumerate}
    \item $\frac{1}{n} \log M \geq R - \delta$
    \item $\overline{e}_{av}(n, \cC, \fI) \leq \epsilon$,\ and
    \item $\frac{1}{n} \log L \leq R_e$.
   \end{enumerate}
   We call 
   \begin{align*}
    \overline{C}^{AV}_{EA}(\fI) := \sup\{R \geq 0: \ R \ \text{ach. rate for EA message transmission over the AVQC} \ \fI \ \text{under av. error criterion} \}
   \end{align*}
   the \emph{average error entanglement-assisted message transmission capacity of the AVQC $\fI$}. 
  \end{definition}
  As in the case of compound quantum channels, the definition for achievable rates regarding the maximal error criterion can be easily guessed. We denote the corresponding capacity by $C_{EA}^{AV}(\fI)$.  
  \begin{remark}
   As opposed to the unassisted case, we abstain from providing separate definitions of the entangle- ment-assisted message transmission capacities in case that random coding procedures are 
   allowed for message transmission. Since additional entanglement can be used for coordinating random coding procedures, the deterministic and random capacities of arbitrarily 
   varying quantum channels match under entanglement assistance. Especially, phenomena as the so-called Ahlswede dichotomy known from classical \cite{ahlswede78} well as quantum \cite{ahlswede13} channel coding scenarios 
   without assistance do not arise in the present context.
  \end{remark}
  The following theorem is the second main result of this paper, and determines the EA classical message transmission capacities of AVQCs. 
  \begin{theorem} \label{theorem:avqc_ea_capacity}
   Let $\fI \subset \cC(\cH_A, \cH_B)$ be a set of c.p.t.p. maps. It holds
    \begin{align}
     \overline{C}_{EA}^{AV}(\fI) \ = \ C_{EA}^{AV}(\fI) \ = \ C_{EA}\left(\conv(\fI)\right) \  
     = \ \underset{\rho \in \cS(\cH_A)}{\sup} \ \underset{\cN \in \conv(\fI)}{\inf} \ I(\rho, \cN). \label{theorem:avqc_ea_capacity_1}
    \end{align}
  \end{theorem}
  As in the case of compound quantum channels, the inequality $\overline{C}_{EA}^{AV}(\fI) \ \geq \ C_{EA}^{AV}(\fI)$ is obvious from the definitions. The remaining statements of Theorem \ref{theorem:avqc_ea_capacity}
  are proven in detail in Section \ref{subsect:prfs_avqc} below.  \newline 
  From the characterization of the entanglement-assisted classical capacity for the AVQC model, we immediately obtain two important structural properties of these capacities. The first one (Corollary 
   \ref{corr:avqc_ea_additivity} below) is additivity of the entanglement-assisted message transmission capacities. The second (Corollary \ref{corr:avqc_cap_stab}) is stability of the entanglement-assisted 
   message transmission capacities of AVQCs under perturbation of the
   generating set. Both of the mentioned assertions are proven in Section \ref{subsect:prfs_avqc}.
   \begin{corollary}[Additivity of the AVQC EA message transmission capacities] \label{corr:avqc_ea_additivity}
   Let $\fI \subset \cC(\cH_A, \cH_B)$, $\fI' \subset \cC(\cH_A', \cH_B')$ be any two sets of c.p.t.p. maps, and $\fI \otimes \fI' := \{\cN \otimes \cN': \ \cN \in \fI,\ \cN' \in \fI'\}$. It holds
   \begin{align}
       \overline{C}_{EA}^{AV}(\fI) + \overline{C}_{EA}^{AV}(\fI') = \overline{C}_{EA}^{AV}(\fI \otimes \fI') = C_{EA}^{AV}(\fI \otimes \fI') = C_{EA}^{AV}(\fI) + C_{EA}^{AV}(\fI')
  \end{align}
  \end{corollary}
  \begin{corollary}[Stability of the EA message transmission capacities] \label{corr:avqc_cap_stab}
  Let $\cH_A, \ \cH_B$ be Hilbert spaces. For each $\epsilon > 0$ exists a $\delta > 0$, such that for any two sets $\fI, \fI' \subset \cC(\cH_A, \cH_B)$ the 
  implication 
  \begin{align}
   D_\Diamond(\fI, \fI') < \delta \ \Rightarrow \ \left|C_{EA}^{AV}(\fI) - C_{EA}^{AV}(\fI')\right| < \epsilon
  \end{align}
   is true. The same statement holds with $C_{EA}^{AV}$ replaced by $\overline{C}_{EA}^{AV}$.
  \end{corollary}
   We demonstrate that opposed to what we stated in the last section in case of compound quantum channels, the average error as well as maximal error entanglement-assisted classical capacities obey a strong converse for each 
   AVQC. We define, for each $n \in \bbmN$, $\lambda \in (0,1)$, $R_e < \infty$
   \begin{align}
    \overline{N}_{EA}^{AV}(n,\fI,R_e,\lambda) 
    &:= \max\{M: \ \exists (n,L,M)-\text{EA code}\  \cC \ \text{for the AVQC} \ \fI \ \text{such that} \ \overline{e}_{av}(n,\cC, \fI) \leq \lambda, \ L \leq 2^{nR_e}\} \label{def:av_max_message_card}
   \end{align}
   A corresponding quantity $N^{AV}_{EA}$ is defined analogously by replacing the average error with the maximal error function. 
    \begin{theorem} \label{theorem:strong_converse}
     Let $\fI \subset \cC(\cH_A, \cH_B)$ be a set of c.p.t.p. maps. For each $R_e < \infty$, $\lambda < 1$, the claims 
     \begin{enumerate}
     \item $\hspace{.3cm} \underset{n \rightarrow \infty}{\limsup}\  \frac{1}{n}\log \overline{N}^{AV}_{EA}(n, \fI, R_e, \lambda) \ \leq  
     \ \underset{\rho \in \cS(\cH_A)}{\sup} \ \underset{\cN \in \conv(\fI)}{\inf} \ I(\rho, \cN)$, and
     \item $\hspace{.3cm} \underset{n \rightarrow \infty}{\limsup}\ \frac{1}{n}\log N^{AV}_{EA}(n, \fI, R_e, \lambda) \ \leq  
     \ \underset{\rho \in \cS(\cH_A)}{\sup} \ \underset{\cN \in \conv(\fI)}{\inf} \ I(\rho, \cN)$
     \end{enumerate}
     hold. i.e. a general strong converse holds to each of both capacities.
    \end{theorem}
    The above claims nearly immediately follow from a combination of the strong converse to the entanglement-assisted message transmission capacities for perfectly known memoryless quantum channels which is known to be valid 
    \cite{bennett14}, \cite{gupta14} and Theorem \ref{theorem:avqc_ea_capacity}. We give a short argument to prove Theorem \ref{theorem:strong_converse} in Section \ref{subsect:prfs_avqc}.
 \end{subsection}
\end{section}

\begin{section}{Proofs} \label{sect:prfs}
 \begin{subsection}{Compound memoryless channels} \label{subsect:prfs_compound}
   In this section, we prove Theorem \ref{theorem:comp_ea_capacity}. The following proposition asserts existence of codes sufficient for proving achievability therein. 
   \begin{proposition} \label{prop:comp_ea_coding}
    Let $\fI := \{\cN_s\}_{s \in S} \subset \cC(\cH_A,\cH_B)$. For each state $\rho \in \cS(\cH_A)$, and $\delta > 0$ there is a number
    $n_0 := n_0(\delta, \rho)$, such that for each $n > n_0$, we find an $(n,L,M)$-EA message transmission code $\cC$ 
    with $L \leq (\dim\cH_A)^n$, and the conditions 
    \begin{enumerate}
     \item $\frac{1}{n} \log M \ \geq \ \underset{s \in S}{\inf} \ I(\rho, \cN_s) - \delta$, and
     \item $e(n, \cC, \fI) \ \leq \ 2^{-nc}$
    \end{enumerate}
    being fulfilled with strictly positive constant $c := c(\delta, \rho)$.
   \end{proposition}
   The proof of Proposition \ref{prop:comp_ea_coding} given below utilizes coding schemes for compound classical-quantum (cq) channels which we borrow from \cite{mosonyi15}. 
   We first provide the definitions necessary to understand the claim of Proposition \ref{prop:comp_cq_codes} below. A classical-quantum channel 
   (with input finite alphabet $\cX$ and output Hilbert space $\cK$) is a map from $\cX$ to the set of density matrices on $\cK$. The memoryless cq channel generated by a cq-channel map
   $W: \cX \rightarrow \cS(\cK)$, is given by the family $\{W^{\otimes n}: n \in \bbmN\}$, where for each $n  \in \bbmN$, the map $W^{\otimes n}: \cX^n \rightarrow \cS(\cK^{\otimes n})$ is defined 
   by 
   \begin{align*}
    W^{\otimes n}(x^n) := W(x_1) \otimes \cdots \otimes W(x_n) && (x^n = (x_1,\dots, x_n) \in \cX^n).
   \end{align*} 
   For each cq channel $W:\ \cX \rightarrow \cS(\cK)$ and probability distribution $p \in \cP(\cX)$ we define the \emph{Holevo quantity} of $(p,W)$ by 
   \begin{align*}
    \chi(p,W) := S\left(\sum_{x \in \cX} \ p(x) \ W(x) \right) - \sum_{x \in \cX} \ p(x) \ S(W(x)).
   \end{align*}
   An $(n, M)$ message transmission code for $W$ is a family $(u_m,D_m)$,
   where $u_1,\dots,u_m \in \cX^n$, and $D_1,\dots,D_M \in \cL(\cK^{\otimes n})$, $0 \leq D_m \leq \bbmeins$, $m \in [M]$, and $\sum_{m=1}^M D_m \leq \bbmeins$. Proposition \ref{prop:comp_cq_codes} 
   states existence of universal capacity approaching codes with super-polynomial decrease of error for compound classical-quantum channels. For more detailed definitions on classical message transmission
   over compound cq channels, the reader is referred to \cite{mosonyi15}.
   \begin{proposition}[\cite{mosonyi15}, Theorem IV.18] \label{prop:comp_cq_codes}
    Let $\fW := \{W_s: \cX \rightarrow \cS(\cK)\}_{s \in S}$ be a set of cq channels, and $p \in \fP(\cX)$ be a probability distribution. 
    For each $\eta > 0$ there exists a number $n_0 := n_0(\eta)$, such that for each $n > n_0$ we find an $(n,M)$ message transmission 
    code $\cC := (u_m, D_m)_{m \in [M]}$ with 
    \begin{enumerate}
     \item $\frac{1}{n} \log M  \geq \underset{s \in S}{\inf} \ \chi(p,W_s) - \eta$, and
     \item $\sup_{s \in S} \ \max_{m \in [M]} \  \tr(D_m^cW_s^{\otimes n}(u_m)) \leq 2^{-n\tilde{c}}$,
    \end{enumerate}
    where $\tilde{c} := \tilde{c}(\eta,p) >0 $ is a constant.
   \end{proposition}
   \begin{proof}
    The statement is the same as in the above cited Theorem IV.18 from \cite{mosonyi15}, except that therein the average instead of the maximal decoding error was considered. That the statement also
    holds with maximal error can be easily shown by standard methods, which we shortly indicate. \newline 
    Chosse for each $n \in \bbmN$, a $\delta_n$-net $S_{n,\delta_n} \subset S$ in $S$, i.e. a subset of $S$ according to Lemma 2.2 in \cite{mosonyi15}. Let $\cC_n := \{u_m,D_m\}_{m=1}^{M_{n}}$ be an $(n,M)$ message 
    transmission code with average error 
    \begin{align}
     \sup_{s \in S} \ \frac{1}{M_n} \  \sum_{m=1}^{M_n} \tr(D_m^cW_s^{\otimes n}(u_m)) \leq 2^{-n\hat{c}}
    \end{align}
    with some constant $\hat{c} > 0$. Clearly,  
    \begin{align}
    \frac{1}{M_n} \  \sum_{m=1}^{M_n} \tr\left(D_m^c \tfrac{1}{|S_n|}\sum_{s \in S_n} W_s^{\otimes n}(u_m)\right) \leq 2^{-n\hat{c}} 
    \end{align}
    also holds. By a standard argument, we can select a subcode $\tilde{\cC}_n = (u_m,D_m)_{m \in \cM'_n}$ of $\cC_n$ with $|\cM'_n| \geq \frac{M_n}{2}$ and maximal error
    \begin{align}
    \sup_{s \in S_n} \sup_{m \in \cM'}  \tr\left(D_m^c W_s^{\otimes n}(u_m)\right) \leq |S_n|\cdot 2^{-n\hat{c}}. 
    \end{align}
    If we now choose $\delta_n := \exp(-n\hat{c}/(2|\cX|\cdot \dim\cK^2))$, $n \in \bbmN$, and set $\tilde{c} = \hat{c}/2$, we have by properties of the sets $S_n$
    \begin{align}
    \sup_{s \in S} \sup_{m \in \cM'}  \tr\left(D_m^c W_s^{\otimes n}(u_m)\right) \leq  2^{-n\tilde{c}}
    \end{align}
    for each large enough $n \in \bbmN$.
   \end{proof}

   The second ingredient to the proof of Proposition \ref{prop:comp_ea_coding} is an encoding construction introduced in \cite{hsieh08} for proving the coding theorem for the entanglement-assisted
   classical capacities of quantum multiple access channels. In Appendix \ref{appendix:encoding} we review the definitions and some properties known from \cite{hsieh08}.
   The following lemma provides universal approximations of certain quantum mutual information quantities arising from the mentioned encoding maps by Holevo quantities of certain effective classical-quantum channels. 
   \begin{lemma} \label{lemma:mutual_approx}
    Let $\cH \simeq \bbmC^d$, $\sigma \in \cS(\cH)$ be a state, and
    \begin{align*}
     \psi := \sum_{i=1}^{d} \sqrt{\alpha_i}\  \gamma_i \otimes \gamma_i,
    \end{align*}
    the Schmidt decomposition of a state vector of a purification of $\sigma$ with $\{\gamma_i: 1 \leq i \leq d\}$ being an orthonormal basis of eigenvectors of $\sigma$, 
    and $\alpha_1,\dots,\alpha_d$ the eigenvalues (counting zero eigenvalues).
    Let $k \in \bbmN$. There is a family $\{\tilde{\cE}_x\}_{x \in \cX} \in \cC(\cH^{\otimes k}, \cH^{\otimes k})$, such that for each Hilbert space 
    $\cK$, and each channel 
    $\cN \in \cC(\cH,\cK)$, with $q_\ast$ being the equidistribution on $\cX$, and the cq channel $V$ being defined by 
    \begin{align*}
     V(x) := \cN^{\otimes k}\circ \tilde{\cE}_{x} \otimes \id_{\cH}^{\otimes k}(\ket{\psi}\bra{\psi}^{\otimes k}) &&(x \in \cX),
    \end{align*}
    the inequality 
    \begin{align*}
     \left| k\cdot I(\sigma, \cN) - \chi(q_\ast, V) \right| \leq 2d \cdot \log(k+1)
    \end{align*}
    holds.
   \end{lemma}
   \begin{proof}
    See Appendix \ref{appendix:encoding}.
   \end{proof}
    \begin{proof}[Proof of Proposition \ref{prop:comp_ea_coding}]
     Define $d_A := \dim \cH_A,\ d_B := \dim \cH_B$. Fix a state $\sigma \in \cS(\cH_A)$, and a number $\delta > 0$, such that 
     \begin{align*}
      \underset{s \in S}{\inf} \ I(\sigma, \cN_s) - \delta > 0, 
     \end{align*}
      otherwise there is nothing to prove. Let $k \in \bbmN$ 
     be large enough to suffice the inequality 
     \begin{align}
      \frac{1}{k+1}\left(2 d_A \log (k+1) + \log(d_A \cdot d_B)\right) \leq \frac{\delta}{2}. \label{k_bound}
     \end{align}
     Let $\sigma := \sum_{i=1}^{d_A} \alpha_i \ket{\gamma_i}\bra{\gamma_i}$ be a spectral decomposition of $\sigma$,  
     \begin{align*}
      \psi := \sum_{i=1}^{d_A} \sqrt{\alpha_i} \ \gamma_i \otimes \gamma_i
     \end{align*}
     be a state vector of a purification of $\sigma$, and set $\Psi := \ket{\psi}\bra{\psi}$. Define, for each $s \in S$, the cq channel $V_s: \cX \rightarrow \cS((\cH_B \otimes \cH_B)^{\otimes k})$, by
     \begin{align*}
      V_s (x) := \cN_s^{\otimes k} \circ \tilde{\cE}_x\otimes \id_{\cH}^{\otimes k}(\Psi^{\otimes k}) &&(x \in \cX), 
     \end{align*}
     where $\{\tilde{\cE}_x\}_{x \in \cX} \subset \cC(\cH_A^{\otimes k}, \cH_A^{\otimes k})$ is a family which fulfills the assertions from Lemma \ref{lemma:mutual_approx}. 
     Let $n > k$ be a blocklength, written as $n = k\cdot a + b$ with $a, b \in \bbmN$, $0 \leq b < k$, 
     and fix $q_\ast$ to be the equidistribution on $\cX$. For each large enough $n$ (and consequently large enough $a$) we find, according to Proposition \ref{prop:comp_cq_codes}, an $(a,M)$ cq message transmission code 
     $\tilde{\cC} := (u_m, \tilde{D}_m)_{m =1}^{M}$ with 
     \begin{align}
      &\frac{1}{a} \log M \geq \inf_{s \in S} \chi(q_\ast, V_s) -\frac{\delta}{2}, \ \text{and} \label{good_cq_codes_ea_rate} \\ 
      &\underset{m \in [M]}{\max} \ \tr(\tilde{D}_m^c V_s^{\otimes a}(u_m)) \leq 2^{-a\tilde{c}} \nonumber 
     \end{align}
     for each $s \in S$ with a constant $\tilde{c} > 0$. Based on the objects introduced, we construct an EA message transmission code for $\fI$, where we assume 
     $\Psi^{\otimes n}$ to be the 
     entanglement resource consumed (at this stage, it is clear that the code constructed will suffice the stated bound on $L$). 
     Define for each message $m \in [M]$, and corresponding codeword $u_m = (u_{m,1},\dots, u_{m,a})$ from $\tilde{\cC}$
     \begin{align*}
      \cE_m := \tilde{\cE}_{u_{m,1}} \otimes \cdots \otimes \tilde{\cE}_{u_{m,a}} \otimes \id_{\cH}^{\otimes b},
     \end{align*}
     and 
     \begin{align*}
     D_m := \tilde{D}_m \otimes \bbmeins^{\otimes b}.
     \end{align*}
     With these definitions, $\cC: = (\Psi^{\otimes n}, \cE_m, D_m)_{m=1}^M$ is an $(n, L, M)$-code for EA message transmission with maximal error 
     \begin{align}
      e(n, \cC, \fI) = \sup_{s \in S} \ \max_{m \in [M]} \tr(D_m^c V_s^{\otimes a}(u_m)) \leq  2^{-a\tilde{c}} \leq 2^{-nc},  \label{prop:comp_ea_coding_final_1}
     \end{align}
     with $c := \tilde{c}/(k+1)$, and rate
     \begin{align}
      \frac{1}{n} \ \log M 
      &\geq \frac{1}{(k+1)}\left( \inf_{s \in s}(\chi(q_\ast, V_s) - \frac{\delta}{2}\right)   \nonumber \\
      &\geq \frac{1}{(k+1)}\left( k \cdot \underset{s \in S}{\inf} \ I(\sigma, \cN_s) - 2 \cdot d_A \cdot \log(k+1) -  \frac{\delta}{2}\right)  \nonumber \\
      &\geq \underset{s \in S}{\inf} \ I(\sigma, \cN_s) - \delta.  \label{prop:comp_ea_coding_final_2}
     \end{align}
     The first of the above inequalities is the one in (\ref{good_cq_codes_ea_rate}), and the second arises from application of Lemma (\ref{lemma:mutual_approx}). 
     The third inequality is by our choice of $k$ from (\ref{k_bound}) together with the trivial bound $I(\sigma, \cN_s) \leq \log d_A \cdot d_B$ on the quantum mutual information. \newline 
     The inequalities in (\ref{prop:comp_ea_coding_final_1}) and (\ref{prop:comp_ea_coding_final_2}) together prove the claim of the proposition. 
    \end{proof}
    Next, we prove the full statement of Theorem \ref{theorem:comp_ea_capacity}. Achievability follows from Proposition \ref{prop:comp_ea_coding}. For proving the weak converse, we invoke the following two
    lemmas. The first one is from \cite{holevo02}
    \begin{lemma}[\cite{holevo02}] \label{lemma:holevo_mutual_ineq}
     Let $\cN \in \cC(\cH_A,\cH_B)$ be a c.p.t.p. map, $\Psi \in \cS(\cK_A \otimes \cK_B)$ a pure state, $\rho := \tr_{\cK_B} \Psi$,  $\cE_1,\dots,\cE_M \in \cC(\cK_A, \cH_A)$, 
     and $q \in \fP([M])$ a probability distribution on $[M]$. It holds
     \begin{align}
      \chi(q,V) \leq I(\overline{\tau},\cN),
     \end{align}
     where $V$ is the cq-channel defined by 
     \begin{align}
      V(m) := \cN \circ \cE_m \otimes \id_{\cK_B}(\Psi),
     \end{align}
     and 
     \begin{align}
      \overline{\tau} := \sum_{i =1}^M \ q(m)  \cdot \cE_m(\rho).
     \end{align}
    \end{lemma}
     The lemma below states subadditivity for the quantum mutual information, originally known to hold from \cite{adami97}. 
     \begin{lemma}[\cite{adami97}]\label{mutual_subadditivity} 
      Let $\rho \in \cS(\cH_1 \otimes \cH_2)$ be a density matrix with $\rho_1, \rho_2$ being the marginals on $\cH_1, \cH_2$ deriving from $\rho$, and let $\cN_i \in \cC(\cH_i, \cK_i)$ be a c.p.t.p. 
      map for $i = 1,2$. 
      It holds
      \begin{align*}
       I(\rho, \cN_1 \otimes \cN_2) \leq I(\rho_1, \cN_1) + I(\rho_2, \cN_2).
      \end{align*}
     \end{lemma}
    With the prerequisites picked up, we are ready for the proof of Theorem \ref{theorem:comp_ea_capacity}.
    \begin{proof}[Proof of Theorem \ref{theorem:comp_ea_capacity}]
     The proof of achievability with the maximal message transmission error under consideration follows directly from Proposition \ref{prop:comp_ea_coding}. 
    It remains, to prove the (weak) converse for $\overline{C}_{EA}$. Let, for an arbitrary, fixed blocklength $n \in \bbmN$, $\cC := (\Phi, \cE_m, D_m)_{m=1}^M$ be any $(n,L,M)$-EA message transmission code for
    $\fI$ with average error $\overline{e}(n,\cC,\fI) := \overline{e}_n < 1$. With $q_{\ast}$ being the equidistribution on $[M]$, define states 
    \begin{align*}
     \rho 					&:= \tr_{\cK_B}\Phi, \hspace{.3cm} \text{and} \\
     \tau					&:= \sum_{m=1}^M q_\ast(m) \cE_{m}(\rho).  
    \end{align*}
    Let, for each $1 \leq i \leq n$, $\tau_i$ be the marginal density matrix deriving from $\tau$ on the $i$-th tensor factor in $\cH_A^{\otimes n}$. 
    Define a cq-channel $V_{s}$ by
    \begin{align*}
     V_{s}(m) := \cN_s^{\otimes n} \circ \cE_{m} \otimes \id_{\cK_B}(\ket{\Phi}\bra{\Phi}) &&(m \in [M]).
    \end{align*}
    for each $s \in S$. It holds
    \begin{align}
     \chi(q_\ast, V_{s}) 
     & \leq I(\tau, \cN_s^{\otimes n})  \leq \sum_{i=1}^n I(\tau_i, \cN_s), \label{thm:comp_ea_capacity_converse_4_einhalb}
    \end{align}
    where the left inequality above is by Lemma \ref{lemma:holevo_mutual_ineq}, and the right inequality follows from $(n-1)$-fold application of Lemma \ref{mutual_subadditivity}.
    Now, let $X$ be the equidistributed random variable on the message set $[M]$ (i.e. $\prob(X = m) = q_{\ast}(m)$ for each $m \in [M]$), and define a conditional probability by
    \begin{align*}
     \prob\left(Y_s = m' | X = m\right) := \tr(D_{m'}V_{s}(m)) &&(m,m' \in [M]).
    \end{align*}
    for each $s \in S$. In each case of $s$, we have
    \begin{align}
     \log M  
     &= H(X)  \nonumber \\
     &= I(X,Y_s) + H(X|Y_{s}) \nonumber \\
     &\leq I(X;Y_{s}) + \overline{e}_n \log M + 1 \nonumber \\
     &\leq \chi(q_\ast, V_{s}) + \overline{e}_n \log M + 1,  \label{thm:comp_ea_capacity_converse_5}
    \end{align}
    where the first inequality above is by Fano's Lemma, and the second is by Holevo's bound \cite{holevo73}. 
    We conclude
    \begin{align}
     \frac{1}{n}\log M \ 
     &\leq  \ \frac{1}{n}(\chi(q_\ast, V_{s}) + \overline{e}_n \log M + 1) \\
     &\leq  \ \frac{1}{n} \sum_{i=1}^n I(\tau_i, \cN_s) + \frac{\overline{e}_n \log M}{n} + \frac{1}{n} \\
     &\leq  \ I(\overline{\tau}, \cN_s) + \frac{\overline{e}_n \log M}{n} + \frac{1}{n}. \label{thm:comp_ea_capacity_converse_6}
    \end{align}
    where the second inequality above is by (\ref{thm:comp_ea_capacity_converse_4_einhalb}), and the third, by concavity of the quantum channel mutual  information in the input state, together with the definition
    \begin{align*}
     \overline{\tau} := \frac{1}{n} \sum_{i=1}^n \tau_i.
    \end{align*}
    Finally, minimizing over $s \in S$ and subsequently maximizing over states in $\cS(\cH_A)$ on the r.h.s. of
    (\ref{thm:comp_ea_capacity_converse_6}), 
    we arrive at the inequality 
    \begin{align*}
     \frac{1}{n}\log M \leq  \sup_{\rho \in \cS(\cH_A)}\inf_{s \in S}I(\rho, \cN_s) + \overline{e}_n \cdot \frac{\log M}{n} + \frac{1}{n}.
    \end{align*}
    Since for each sequence of $(n,L,M)$-codes with average errors $\overline{e}_n \rightarrow 0 \ (n\rightarrow \infty)$, the remainder terms vanish, the converse holds.
    \end{proof}
    \begin{remark}
    As can be noticed by inspection of the proof of the converse part in Theorem \ref{theorem:comp_ea_capacity} above, even dropping the condition of principal rate-boundedness on the dimensions of the entanglement resource 
    would not lead to higher communication rates. On the other hand, the proof of the achievability part shows that memoryless extensions of a pure entangled 
    state on $\cH_A \otimes \cH_A$ always suffice as a resource to achieve the capacity.
    \end{remark}
We conclude this section by demonstrating that for both capacities a general strong converse fails to hold. The following example of a compound quantum channel without an average error EA message transmission capacity 
is inspired from \cite{ahlswede67}. 
\begin{example} \label{nostrongconverse_example}
 There is a set $\fI = \{\cN_1, \cN_2\} \subset \mathcal{C}(\mathbbm{C}^5, \mathbbm{C}^5)$ such that for each $n$ 
 \begin{align}
  \frac{1}{n}\log \overline{N}_{EA}(n,\fI, 1, \tfrac{1}{2})  > \sup_{\rho \in \cS(\cC^5)} \ \min_{i=1,2} I(\rho, \cN_i)
 \end{align}
 holds.
\end{example} \label{comp_av_str_conv_counterexample}
In the following, we present the example we stated to exist. 
Define the set $\fI:=\{\mathcal{N}_1,\mathcal{N}_2\}\subset \mathcal{C}(\mathbbm{C}^5, \mathbbm{C}^5)$ formed by the entanglement breaking channels
\begin{align}
 \mathcal{N}_1(a) &:= \sum_{i=1}^2 tr(E_{ii}a) E_{ii} + \sum_{j=3}^5 tr(E_{jj}a) E_{33} &(a \in \mathcal{L}(\mathbbm{C}^5))\\
 \mathcal{N}_2(a) &:= \sum_{i=4}^5 tr(E_{ii}a) E_{ii} + \sum_{j=1}^3 tr(E_{jj}a) E_{33} &(a \in \mathcal{L}(\mathbbm{C}^5)),
\end{align}
where we used the shortcuts $E_{ij} := \ket{e_i} \bra{e_j}$ for $i,j \in \{1,\dots,5\}$ with an orthonormal basis $\{e_i\}_{i=1}^5 \subset \mathbbm{C}^5$ for the matrix units. 
One can show that 
\begin{align}
 \max_{\rho \in \mathcal{S}(\mathbbm{C}^5)} I(\rho,\mathcal{N}_1) = \max_{\rho \in \mathcal{S}(\mathbbm{C}^5)} I(\rho,\mathcal{N}_2) = \log3 \label{nostrongconverse_1}
\end{align}
holds, which for the channel $\mathcal{N}_1$ ($\mathcal{N}_2$) is attained on the set 
\begin{align}
 A_1 &:=  \{\rho:\ \braket{e_1,\rho e_1} = \braket{e_2, \rho e_2} = \tfrac{1}{3}\} \\
 A_2 &:=  \{\rho:\ \braket{e_4,\rho e_4} = \braket{e_5, \rho e_5} = \tfrac{1}{3}\}
\end{align}
and nowhere else. Consequently, the sets $A_1$ an $A_2$ are non-intersecting, and therefore
\begin{align}
  \max_{\rho \in \mathcal{S}(\mathbbm{C}^5)} \min_{i = 1,2} I(\rho,\mathcal{N}_i) < \log3 \label{nostrongconverse_2}
\end{align}
holds. We now show, by constructing sufficient codes, the inequality 
\begin{align}
 \overline{N}_{EA}(n,\fI, 1, \frac{1}{2}) \ \geq \  2 \cdot 3^n - 1,
\end{align}
which, together with (\ref{nostrongconverse_2}) leads to a contradiction to the inequality from (\ref{average_error_compound_strong_converse}) in this case.
For fixed $n \in \mathbbm{N}$, define $e := \bigotimes_{i=1}^n e_1$, 
\begin{align}
V &:= \{1,2,3\}^n \cup \{3,4,5\}^n \\
e_v &:= e_{v_1} \otimes \dots \otimes e_{v_n}, \\
 \mathcal{E}_{v}(\cdot) &:=  \ket{e_v}\bra{e}(\cdot)\ket{e}\bra{e_v}  \ \text{and} \  D_v := \ket{e_v}\bra{e_v} \otimes \mathbbm{1}_{\mathbbm{C}^5}
\end{align}
for each $v := (v_1,...,v_n) \in V$. Then $\cC: = (\Phi^{\otimes n}, \mathcal{E}_v, D_v)_{v \in V}$ with $\Phi := \ket{\phi}\bra{\phi}$, $\phi := e_1 \otimes e_1$is an $(n, 1,|V|)$-code for entanglement-assisted message transmission with
$|V| = 2\cdot 3^n - 1$ and average error bounded by 
\begin{align}
 \max_{i = 1,2} \overline{e}(n, \cC, \mathcal{N}_i^{\otimes n}) = \frac{3^n}{2\cdot 3^n - 1} < \frac{1}{2}.
\end{align}
Consequently, the channel has no strong converse property for the average error criterion.\newline
\begin{remark}
 The reader may notice that the above introduced set $\fI$ is also an explicit example of a compound quantum channel, where the users have to pay a price in capacity for not knowing
 the channel. Combination of (\ref{nostrongconverse_1}) and (\ref{nostrongconverse_2}) together with the coding theorem for the entanglement-assisted capacities for perfectly 
 known memoryless quantum channels leads to
 \begin{align}
    C_{EA}(\fI) = \underset{\rho \in \mathcal{S}(\mathbbm{C}^5)}{\max} \underset{i=1,2}{\min} I(\rho, \mathcal{N}_i) <   
    \underset{i=1,2}{\min} \underset{\rho \in \mathcal{S}(\mathbbm{C}^5)}{\max} I(\rho, \mathcal{N}_i) = \underset{i= 1,2}{\min}\ C_{EA}(\cN_i).
 \end{align}
\end{remark}
It remains to give evidence to the claim that also no general strong converse does hold to the maximal-error EA classical capacity of compound quantum channel. 
It becomes apparent from the following lemma that there is no essential difference between the average error and maximal error criterion, if the users are supplied with free rate-bounded entanglement assistance. 
\begin{lemma}\label{lemma:comp_av_max_equal}
 For each set $\fI \subset \cC(\cH_A, \cH_B)$ there is a number $R_e$, such that for each $n$ and each $\lambda \in (0,1)$, 
 \begin{align}
  \overline{N}_{EA}(n,\fI,R_e,\lambda) = N_{EA}(n,\fI,R_e,\lambda).
 \end{align}
\end{lemma}
\begin{proof}
 The inequality $\overline{N}_{EA}(n,\fI,R,\lambda) \geq N_{EA}(n,\fI,R,\lambda)$ is obvious for each $0 \leq R < \infty$ from the definitions. 
 We prove that the reverse inequality $N_{EA}(n,\fI,R,\lambda) \geq \overline{N}_{EA}(n,\fI,R,\lambda)$ does hold for large enough $R$. 
 Let $\tilde{\cC} := (\tilde{\Psi}, \tilde{\cE}_m, \tilde{D}_m)_{m = 1}^M$ be an 
 $(n,\tilde{L},M)$-EA message transmission code with average error 
 \begin{align}
  \overline{e}(n,\tilde{\cC},\fI) := \overline{\lambda} \in (0,1). 
 \end{align}
 We show that there exists an $(n, L, M)$-EA message transmission code $\cC := (\Psi, \cE_m, D_m)$ whose maximal error equals the average error of $\tilde{\cC}$, i.e. 
 \begin{align}
  e(n, \cC, \fI) = \overline{e}(n, \tilde{\cC}, \fI),
 \end{align}
 which implies, by maximization the desired inequality. Let $\{\sigma_1,\dots,\sigma_M\} \subset \fS_M$ be the set of cyclic translations on $[M]$. i.e. 
 \begin{align}
  \sigma_i(m) := m \oplus i   &&(m, i \in [M]),
 \end{align}
 where $\oplus$ is the modulo-M addition defined on $[M]$. It is clear that 
 \begin{align}
  \frac{1}{M} \sum_{m=1}^M \tr \left(\tilde{D}^c_m(\cN^{\otimes n}\circ \tilde{\cE}_{m} \otimes \id(\tilde{\Psi}))\right) 
  \ = \ \frac{1}{M} \sum_{i=1}^M \tr \left(\tilde{D}^c_{\sigma_i(m')}(\cN^{\otimes n}\circ \tilde{\cE}_{\sigma_i(m')} \otimes \id(\tilde{\Psi}))\right)
 \end{align}
 for each $m' \in [M],\ \cN \in \fI$. We define the components of the code $\cC$. Let $\Psi := \tilde{\Psi} \otimes \Phi$, where $\Phi := \ket{\phi}\bra{\phi}$ is the maximally entangled state on $\tilde{\cK}_A \otimes \tilde{\cK}_B$
 with $\tilde{\cK}_A = \tilde{\cK}_B = \bbmC^M$, and 
 \begin{align}
  \phi := \frac{1}{\sqrt{M}} \sum_{k=1}^M e_k \otimes e_k. 
 \end{align}
  We define 
  \begin{align}
   \cE_m(a) &:= \sum_{k=1}^M \ \tilde{\cE}_{\sigma_k(m)} \circ \tr_{\tilde{\cK}_A}(\bbmeins_{\cK_A} \otimes \ket{e_k}\bra{e_k} a ) &&(a \in \cL(\cK_A \otimes \tilde{\cK}_A), m \in [M]), \\
    D_m   &:= \sum_{k=1}^M \tilde{D}_{\sigma_k(m)} \otimes \ket{e_k}\bra{e_k}    &&(m \in [M]).
   \end{align}
  With these definitions, it holds for each $m \in [M]$, $\cN \in \fI$ 
  \begin{align}
   \tr\left(D^c_m\left(\cN^{\otimes n}\circ \cE_m \otimes \id_{\cK_B}(\Psi)\right) \right) 
   &= \frac{1}{M} \sum_{i=1}^M \tr \left(\tilde{D}^c_{\sigma_i(m)}(\cN^{\otimes n}\circ \tilde{\cE}_{\sigma_i(m)} \otimes \id(\tilde{\Psi}))\right) \\
   &= \overline{\lambda}.
  \end{align}
  Maximizing both sides of the above inequality over all $m \in [M]$ shows that the maximal error of $\cC$ equals the average error of $\tilde{\cC}$ for each given channel $\cN$. Maximizing over all channels in $\fI$ 
  proves our claim. 
\end{proof}
\end{subsection}

\begin{subsection}{Arbitrarily varying channels} \label{subsect:prfs_avqc}
In this section, we give a full proof of Theorem \ref{theorem:avqc_ea_capacity}. The proof of achievability is performed in two steps. In Lemma \ref{lemma:finite avqc_coding} 
below, 
we show that sufficient maximal-error codes exist for each large enough blocklength for each AVQC which is generated by a finite set $\fI$ of quantum channels. Afterwards, we derive sufficient codes
for each given AVQC (not necessarily generated by a finite or countable set of quantum channels). The strategy of proof in this case is, to combine codes
for finite AVQCs with suitable approximations of arbitrary AVQCs by finite AVQCs. For proving the coding result for finite AVQCs, we use Ahlswede's robustification lemma, which we state first.
\begin{theorem}[Robustification technique, cf. Theorem 6 in Ref. \cite{ahlswede86}]\label{robustification-technique}\ \\
Let $S$ be a set with $|S|<\infty$ and $n \in\bbmN$. If a function $f:S^n\to [0,1]$ satisfies
\begin{align}
 \sum_{s^n\in S^n}f(s^n)q(s_1)\cdot\ldots\cdot q(s_n)\geq 1-\gamma \label{eq:robustification-1}
\end{align}
for each type $q$ of sequences in $S^n$ for some $\gamma\in [0,1]$, then
\begin{align}
  \frac{1}{n!}\sum_{\sigma\in \fS_n}f(\sigma(s^n))\ge 1-(n+1)^{|S|}\cdot \gamma\qquad \forall s^n \in S^n.\label{eq:robustification-2}
\end{align}
\end{theorem}
The following lemma states existence of codes, sufficient for proving the achievability part of Theorem \ref{theorem:avqc_ea_capacity}.
 \begin{lemma}\label{lemma:finite avqc_coding}
 Let $\fI := \{\cN_s\}_{s \in S} \subset \cC(\cH_A,\cH_B)$ be a finite set of c.p.t.p. maps, and define $\tilde{\cN}_p := \sum_{s \in S} p(s) \cN_s$ for each probability distribution $p \in \fP(S)$. For each 
 $\delta > 0$ exists a number $n_0(\delta)$, such that we find for each $n > n_0$ an $(n,L,M)$-EA message transmission code $\cC$ with 
 \begin{enumerate}
  \item $e_{av}(n,\cC, \fI) \leq 2^{-n\hat{c}}$ with a constant $\hat{c}(\fI,\delta) > 0$, 
  \item $\frac{1}{n} \log M \geq \sup_{\rho \in \cS(\cH_A)}\ \inf_{p \in \fP(S)} \ I(\rho, \tilde{\cN}_p) - \delta$, and
  \item $\frac{1}{n} \log L \leq \dim \cH_A + \delta$.
 \end{enumerate}
 \end{lemma}
 \begin{proof}
  Applying Proposition \ref{prop:comp_ea_coding} on the compound channel generated by the set $\conv(\fI)$ of c.p.t.p. maps, we obtain, provided $n$ is large enough, an $(n,\tilde{L},M)$-EA message transmission code
  $\tilde{\cC}=(\Psi, \tilde{\cE}_m, \tilde{D}_m)_{m=1}^M$ with 
  \begin{align}
   \frac{1}{n} \log M &\geq \underset{\rho \in \cS(\cH_A)}{\sup} \ \underset{p \in \fP(S)}{\inf} \ I(\rho, \tilde{\cN}_p) - \delta,  \label{prf:fin_av_rate}\\
   e(n,\tilde{\cC}, \conv(\fI)) &\leq 2^{-nc} =: \tilde{\epsilon}_n,   \label{prf:fin_av_comp_err}
  \end{align}
  and $\tilde{L} \leq \dim (\cH_A)^{\otimes n}$. We define a function $f_m: S^n \rightarrow [0,1]$, $m \in [M]$ by 
  \begin{align}
   f_m(s^n) := \tr\left(\tilde{D}_m\left(\cN_{s^n}\circ \tilde{\cE}_m \otimes \id_{\cK_B}(\Psi)) \right)\right).  \label{prf:fin_av_rob_def}
  \end{align}
  From (\ref{prf:fin_av_comp_err}) we infer
  \begin{align*}
   \sum_{s^n \in S^n} p(s_1)\cdot \dots \cdot p(s_n) f_m(s^n) \geq 1 - \tilde{\epsilon}_n
  \end{align*}
   for each $p \in \fP(S)$, $m \in [M]$. Define, for each $\sigma \in \fS_n$, and $\alpha \in \fS_M$ an $(n,\tilde{L},M)$-EA message transmission code $\tilde{\cC}_{\alpha,\sigma} := (\Psi, \tilde{\cE}_{m,\alpha, \sigma}, 
   \tilde{D}_{m,\alpha,\sigma})_{m=1}^M$ with 
   \begin{align*}
    \tilde{\cE}_{m,\alpha,\sigma}(a) := U_{A,\sigma} \tilde{\cE}_{\alpha(m)}(a) U_{A,\sigma}^\ast &&(a \in \cL(\cH_{A}^{\otimes n})
   \end{align*} 
   and 
   \begin{align*}
    \tilde{D}_{m,\alpha,\sigma} := (U_{B,\sigma} \otimes \bbmeins_{\cK_B})\tilde{D}_{\alpha(m)}(U_{B,\sigma} \otimes \bbmeins_{\cK_B})^\ast
   \end{align*}
   for each $m \in [M]$, where $U_{A,\sigma}$, $U_{B,\sigma}$ are the unitaries on $\cH_A^{\otimes n}$ resp. $\cH_B^{\otimes n}$ permuting the tensor factors according to $\sigma$, i.e. 
   \begin{align*}
    U_{X,\sigma} \ x_1 \otimes \dots \otimes x_n := x_{\sigma(1)} \otimes \dots \otimes x_{\sigma(n)} &&(x_1,\dots,x_n \in \cH_X, \sigma \in \fS_n)
   \end{align*}
   with $X = A, B$. It then holds with the definitions given 
   \begin{align*}
    f_{\alpha(m)}(\sigma(s^n)) = \tr\left(\tilde{D}_{\alpha, \sigma, m}\left(\cN_{s^n}\circ \tilde{\cE}_{\alpha, \sigma, m} \otimes \id_{\cK_B}(\Psi)) \right)\right)
   \end{align*}
   for each $\sigma \in \fS_n, m \in [M], \ s^n \in S^n$. From (\ref{prf:fin_av_comp_err}), and (\ref{prf:fin_av_rob_def}), together with application of Theorem \ref{robustification-technique} and the fact that 
   permutations of the messages do not change the maximal error of the code,
   we conclude that
   \begin{align}
    1 - \frac{1}{n! M!} \sum_{\sigma \in \fS_n}\sum_{\alpha \in \fS_M} f_{\alpha(m)}(\sigma(s^n))\  \leq \  (n+1)^{|S|} \cdot \tilde{\epsilon}_n =: \epsilon_n \label{prf:av_fin_rcod}
    \end{align}
   holds for each $s^n \in S^n$. Let $X_1,\dots,X_K$ be a sequence of i.i.d. random variables each equidistributed on $\fS_n$, and $Y_1,\dots, Y_K$ be a sequence of i.i.d. random variables equidistributed on $\fS_M$ with 
   $K := \lceil 2^{n\frac{c}{8}} \rceil$. Define 
   $g_m(\alpha, \sigma ,s^n) := 1 - f_{\alpha(m)}(\sigma(s^n))$ for each $\sigma \in \fS_n, \alpha \in \fS_M, s^n \in S^n$. If we choose $\sigma$, and $\alpha$ randomly according to $(X_j, Y_j)$ the expectation equals the l.h.s. 
   of (\ref{prf:av_fin_rcod}), which implies 
   \begin{align}
    \bbmE\left[g_m(X_j, Y_j, s^n)\right] \leq \epsilon_n \label{prf:fin_av_expect}
   \end{align}
   for each $m \in [M], s^n \in S^n$. It holds for each $\nu > \epsilon_n$, $s^n \in S^n$, and $\alpha > 0$
   \begin{align}
    \prob\left(\sum_{k=1}^K g_m(X_k, Y_k,s^n) > K \nu \right) 
    &= \prob\left(\prod_{k=1}^K \exp(\alpha \cdot g_m(X_k, Y_k, s^n)) > 2^{\alpha K  \nu}\right) \nonumber \\
    &\leq 2^{-\alpha K\nu} \cdot \prod_{k=1}^{K}\bbmE\left[\exp(\alpha \cdot g_m(X_k, Y_k ,s^n))\right],
  \end{align}    
    by Markov's inequality. By convexity of the exponential function, it holds $2^{\alpha x} \leq (1-x) 2^{\alpha \cdot 0} + x \cdot 2^{\alpha \cdot 1} \leq 1 + 2^\alpha \cdot x$ for each $x \in [0,1]$. Consequently 
    \begin{align}
     \bbmE\left[\exp(\alpha \cdot g_m(X_k, Y_k ,s^n))\right] \leq 1 + 2^{\alpha} \cdot  \bbmE g_m(X_k, Y_k ,s^n) \leq 1 + 2^\alpha \epsilon_n
    \end{align}
    holds for each $k \in [K], m \in [M]$. We have
    \begin{align}
    \prob\left(\sum_{k=1}^K g_m(X_k, Y_k,s^n) > K \nu \right) 
    &\leq 2^{-\alpha K\nu} \cdot (1 + 2^\alpha \epsilon_n)^K \\
    &\leq 2^{-K(\alpha \nu - \log(1 + 2^\alpha \epsilon_n)) } \\   
    &\leq 2^{-K(\alpha \nu - 2 \cdot 2^\alpha \epsilon_n)) } \label{prf:av_fin_prob_bound}
    \end{align}
    where the last inequality holds for each large enough $n$ by $\log(1+x) \leq 2x$ which is true for all $x \in [0,1]$.
    From (\ref{prf:av_fin_prob_bound}), the choice $\alpha = 1$ and de Morgan's laws, we conclude that
    \begin{align}
     \prob\left(\forall s^n \in S^n, m \in [M]: \ \frac{1}{K}\sum_{k=1}^K \ \tr\left(\tilde{D}_{Y_k, X_k, m}\left(\cN_{s^n}\circ \tilde{\cE}_{Y_k, X_k, m} \otimes \id_{\cK_B}(\Psi)) \right)\right)  \leq \nu \right) 
     \geq 1 -  |S|^n \cdot M \cdot 2^{-K(\nu - 4 \epsilon)} \label{prf:av_fin_prob_bound_x}
    \end{align}
   holds. If we set $\nu := 2^{-n\hat{c}}$, $\hat{c} := c/4$, and choose $n$ large enough, $2^{-K(\nu -4 \epsilon_n)}$ does grow super-exponentially with $n$, therefore the r.h.s. of (\ref{prf:av_fin_prob_bound_x}) is strictly positive 
   by our choice of $K$. Consequently, we find a family $\{(\sigma_1, \alpha_1),\dots, (\sigma_K, \alpha_K)\}$, such that 
   \begin{align}
    \frac{1}{K} \sum_{k=1}^K \ e_{av}(n,\tilde{\cC}_{\alpha_k,\sigma_k}, \cN_{s^n}) \leq 2^{-n\hat{c}}
   \end{align}
   holds.  For each large enough blocklength $n$, we define an $(n,L,M)$-EA message transmission code $\cC := (\Psi \otimes \Phi, \cE_m, D_m)_{m=1}^M$ with $\Phi$ being the maximally entangled state with Schmidt vector
   \begin{align}
    \phi := \frac{1}{\sqrt{K}} \sum_{k=1}^K e_k \otimes e_k
   \end{align}
   on $\tilde{\cK}_A \otimes \tilde{\cK}_B$, $\tilde{\cK}_A = \tilde{\cK}_B = \bbmC^K$, 
   \begin{align}
    \cE_m(a) := \sum_{k=1}^K \tilde{\cE}_{m,\alpha_k,\sigma_k} \circ \tr_{\tilde{\cK}_A}\left((\bbmeins_{\cK_A} \otimes \ket{e_k}\bra{e_k})a\right) &&\left(a \in \cL(\cK_A \otimes \tilde{\cK}_A)\right),
   \end{align}
    and 
    \begin{align}
     D_m := \sum_{k=1}^K \tilde{D}_{m, \alpha_k, \sigma_k} \otimes \ket{e_k} \bra{e_k}
    \end{align}
    for each $m \in [M]$. By definition of $\cC$, we have
    \begin{align}
     e_{av}(n, \cC, \cN_s) = \frac{1}{K} \sum_{k=1}^K e_{av}(n, \tilde{\cC}_{\alpha_k,\sigma_k}, \cN_{s^n}) \leq 2^{-n\frac{c}{2}}.
    \end{align}
    The bound on $L$ is easily verified, it holds, by construction $L = \tilde{L} \cdot K$, and therefore
    \begin{align}
    \frac{1}{n}\log L =  \dim \cH_A + \frac{nc}{8n},
    \end{align}
    which clearly verifies the bound on $L$, if $n$ is large enough.
   \end{proof}
 Next we drop the condition of finiteness for the set $\fI$ generating the AVQC. We show existence of codes for each large enough blocklength, suitable to show the achievability part in 
 Theorem \ref{theorem:avqc_ea_capacity}. This will be done by using codes as derived in Lemma \ref{lemma:finite avqc_coding} for a suitable approximation of $\fI$ by a finite AVQC. Explicitly, such an
 approximation is obtained, by approximating $\conv(\fI)$. For this reason, we first introduce some notions and results from convex geometry. \newline
 For a subset $A$ of a normed space $(V,\|\cdot\|)$, $\overline{A}$ is the closure and $\aff A$ is the affine 
 hull of $A$. If $A$ is a 
 convex set, the relative interior $\ri A$ is the interior and the relative boundary $\rebd A$ of $A$ 
 are the interior and boundary of $A$ regarding the topology on $\aff A$ induced by $\|\cdot\|$. 
  \begin{lemma}[Ref. \cite{ahlswede13}, Lemma 34]\label{haussdorff_1}
  Let $A$, $B$ be compact sets in  $\bbmC^n$ with $A \subset B$ and
  \begin{align}
   d_H(\rebd B, A) = t > 0,
  \end{align}
  where $d_H$ is the Hausdorff distance induced by any norm $\| \cdot \|$ on $\bbmC^n$. Let $P$ a polytope with $A \subset P$ and $d_H(A,P) \leq 
  \delta$, where $\delta \in (0,t]$. Then 
  $P' := P \cap \aff A$ is also a polytope and $P \subset B$. 
 \end{lemma}
 
 \begin{lemma}\label{lemma:arb_av_coding}
  Let $\fI := \{\cN_s\}_{s \in S} \subset \cC(\cH_A, \cH_B)$  be a set of c.p.t.p. maps. For each $\delta > 0$, $\epsilon > 0$, there exists a number $n_0 := n_0(\epsilon, \delta)$, such that
  for each $n > n_0$ there is an $(n,L,M)$-EA message transmission code $\cC$ fulfilling 
  \begin{enumerate}
   \item $e_{av}(n,\cC, \fI) \leq 2^{-n c}$ with a constant $c > 0$,
   \item $\frac{1}{n} \log M \geq \sup_{\rho \in \cS(\cH_A)} \ \inf_{\tilde{\cN} \in \overline{\conv(\fI)}} I(\rho, \tilde{\cN}) - \delta$, and
   \item $\frac{1}{n} \log L \leq \dim \cH_A + \delta$.
  \end{enumerate}
 \end{lemma}

 \begin{proof}
  Let $\delta > 0$ be a fixed number, such that 
  \begin{align}
   \sup_{\rho \in \cS(\cH_A)} \ \inf_{\tilde{\cN} \in \overline{\conv(\fI)}} I(\rho, \tilde{\cN}) - \delta > 0
  \end{align}
  holds (otherwise the lemma is trivial). For technical reasons, it is necessary that $\overline{\conv(\fI)}$ does not touch the boundary of $\cC(\cH_A, \cH_B)$. To also cover situations, where this is not the case, let 
  for $\gamma \in [0,1)$, $\cD_\gamma$ be the $\gamma$-depolarizing channel on $\cL(\cH_A)$ defined by 
  \begin{align}
   \cD_\gamma(x) := (1-\gamma)x + \gamma \cdot \tr(x) \frac{\bbmeins_{\cH_A}}{\dim \cH_A} &&(x \in \cL(\cH_A)).
  \end{align}
  Since $\cD_\gamma \circ \cN  \notin \rebd \cC(\cH_A, \cH_B)$ holds for each $\cN \in \cC(\cH_A, \cH_B)$, $\gamma \in (0,1)$, we have by Lemma 2.3.3 in \cite{webster94} 
  \begin{align}
   \cD_\gamma(\overline{\conv(\fI)}) \subsetneq \ri \cC(\cH_A, \cH_B),
  \end{align}
  which implies, via the obvious relation $\cD_\gamma(\overline{\conv(\fI)}) = \overline{\cD_\gamma(\conv(\fI))}$, positive distance of $\cD_\gamma(\overline{\conv(\fI)})$ to the relative boundary of $\cC(\cH_A, \cH_B)$, i.e.
  \begin{align}
   \min\left\{\|\cN - \cN'\|_{\Diamond}:\ \cN \in \cD_\gamma(\overline{\conv(\fI)}),\ \cN \in \rebd\cC(\cH_A, \cH_B)\right\} \ > \ 0.
  \end{align}
  With Lemma \ref{haussdorff_1} and Theorem 3.1.6 from \cite{webster94}, we find a polytope $\fF_\gamma$ such that $\cD_\gamma(\overline{\conv(\fI)}) \subset \fF_\gamma \subsetneq \cC(\cH_A, \cH_B)$, and moreover 
  \begin{align}
   D_{\Diamond}(\cD_\gamma(\overline{\conv(\fI)}), \fF_\gamma) \leq 2 \eta \label{polytope_approx_distance}
  \end{align}
  holds. Let 
  $\fE_\gamma := \{\tilde{\cN}_e\}_{e \in E_\gamma}$ be the (finite) set of extremal elements of $\fF_\gamma$. 
  If $n$ is large enough, we find, according to Lemma \ref{lemma:finite avqc_coding}, an $(n,L,M)$-EA message transmission code $\tilde{\cC} := (\Psi, \cE_m, \tilde{D}_m)_{m=1}^M$ for the AVQC generated by $\fE_\gamma$, which fulfills
  \begin{align}
   e_{av}(n,\tilde{\cC}, \fE_\gamma) \leq 2^{-n c}, \label{prf:arb_av_coding_perf}
  \end{align}
   $c > 0$, and 
   \begin{align}
    \frac{1}{n} \log M \geq \sup_{\rho \in \cS(\cH_A)} \ \inf_{\tilde{\cN} \in \fF_\gamma} \ I(\rho,\tilde{\cN}) - \frac{\delta}{2}.  \label{prf:arb_av_coding_rate}
   \end{align}
   Each member of $\cD_\gamma(\conv(\fI))$ can be written as a convex combination of elements of $\fE_\gamma$, i.e. 
   \begin{align}
    \cD_\gamma \circ \cN_s  = \sum_{e \in E_\gamma} \ q(e|s) \tilde{\cN}_e
   \end{align}
   with a probability distribution $q(\cdot|s)$ on $E_\gamma$ for each $s \in S$. Define, based on the objects from $\tilde{\cC}$ an $(n,L,M)$-code $\cC := (\Psi, \cE_m, D_m)_{m=1}^M$ with the definition
   $D_m := (\cD_\gamma^\ast \otimes \id_{\cK_B})^{\otimes n}(\tilde{D}_m)$ for each $m \in [M]$, where $\cD_\gamma^{\ast}$ denotes the unital Hilbert-Schmidt adjoint of $\cD_\gamma$. 
   It holds, for each $s^n \in S^n$, $m \in [M]$,
   \begin{align}
    \tr\left(D_m\left(\cN_ {s^n} \circ \cE_m \otimes \id_{\cK_B} (\Psi)\right)\right)
    &= \tr\left(\tilde{D}_m\left(\cD_\gamma^{\otimes n} \circ \cN_ {s^n} \circ \cE_m \otimes \id_{\cK_B} (\Psi)\right)\right) \\
    &= \tr\left(\tilde{D}_m\left(\bigotimes _{i=1}^n  \cD_\gamma \circ \cN_{s_i} \circ \cE_m \otimes \id_{\cK_B} (\Psi)\right)\right) \\
    &= \tr\left(\tilde{D}_m\left(\bigotimes _{i=1}^n  \sum_{e_i \in E_\gamma} q(e_i|s_i) \tilde{\cN}_{e_i} \circ \cE_m \otimes \id_{\cK_B} (\Psi)\right)\right) \\
    &= \sum_{e^n \in E_\gamma^n} q^n(e^n|s^n) \ \tr\left(\tilde{D}_m \left(\tilde{\cN}_{e^n} \circ \cE_m \otimes \id_{\cK_B}(\Psi)\right)\right). 
    \end{align}
    Rearranging and maximizing both sides of the equality above over all messages $m \in [M]$, we infer using (\ref{prf:arb_av_coding_perf})
    \begin{align}
       e_{av}(n,\cC, \fI) \leq e_{av}(n, \tilde{\cC}, \fE_\gamma) \leq 2^{-nc}. 
    \end{align}
    On the other hand, by (\ref{polytope_approx_distance}) in combination with Fannes' inequality, it holds
    \begin{align}
     \sup_{\rho \in \cS(\cH_A)} \ \inf_{\tilde{\cN} \in \fF_\gamma} \ I(\rho,\tilde{\cN}) \geq \sup_{\rho \in \cS(\cH_A)} \ \inf_{\tilde{\cN} \in \overline{\conv(\fI)}} \ I(\rho,\tilde{\cN}) - f(2\gamma).
    \end{align}
    with a function $f: [0,1] \rightarrow [0,1]$, $f(x) \rightarrow 0 \ (x \rightarrow 0)$. Choosing $\gamma > 0$ small enough, it holds with (\ref{prf:arb_av_coding_rate})
    \begin{align}
     \frac{1}{n}\log M \geq \inf_{\tilde{\cN} \in \overline{\conv(\fI)}} \ I(\rho,\tilde{\cN}) - \delta.
    \end{align}
    We are done.
  \end{proof} 
  \begin{proof}[Proof of Theorem \ref{theorem:avqc_ea_capacity}.]
   The inequality 
   \begin{align}
    C_{EA}^{AV}(\fI) \geq \sup_{\rho \in \cS(\cH_A)} \ \inf_{\tilde{\cN} \in \overline{\conv(\fI)}} I(\rho, \tilde{\cN}) 
   \end{align}
   follows directly from Lemma \ref{lemma:arb_av_coding}. The converse inequality is also obvious, since $C_{EA}^{AV}(\fI) \leq C_{EA}(\conv(\fI))$.
  \end{proof}
 \end{subsection}
  Directly from the characterization of the EA message transmission capacities in (\ref{theorem:avqc_ea_capacity_1}), we obtain the following two
  corollaries.
  \begin{proof}[Proof of Corollary \ref{corr:avqc_ea_additivity}]
    The inequality $\overline{C}_{EA}^{AV}(\fI) + \overline{C}_{EA}^{AV}(\fI') \leq \overline{C}_{EA}^{AV}(\fI \otimes \fI')$ is obviously true from the operational definition of the capacities. To  show the reverse 
    inequality, we first consider two  finite sets $\fI := \{\cN_s\}_{s \in S}$, $\fI' := \{\cN_{s'}\}_{{s'} \in S'}$, $|S|, |S'| < \infty$. For each $p \in \fP(S \times S')$ with marginal distributions 
    $q \in \fP(S)$, $q' \in \fP(S')$, we define 
    \begin{align}
     \overline{\cM}_p &:= \sum_{(s,s') \in S \times S'} \ p(s,s') \ \cN_s \otimes \cN'_{s'}, \\
     \overline{\cN}_q &:= \sum_{s \in S} \ q(s) \ \cN_s, \\
     \overline{\cN}'_{q'} &:= \sum_{s' \in S'} \ q'(s') \ \cN'_{s'}.
    \end{align}
    Subadditivity of the quantum mutual information (Lemma \ref{mutual_subadditivity}) then implies for each state $\rho \in \cS(\cH_A \otimes \cH'_A)$ the inequality 
    \begin{align}
     I(\rho, \overline{\cM}_p) \ \leq \ I(\sigma, \overline{\cN}_q) + I(\sigma', \overline{\cN}_{q'}), \label{prf:avqc_add}
    \end{align}
    where $\sigma := \tr_{\cH'_A}\rho$, $\sigma' := \tr_{\cH_A}\rho$ are the marginals of $\rho$. Minimizing both sides of (\ref{prf:avqc_add}) over all probability distributions on $S \times S'$, it holds
     \begin{align}
     \inf_{p \in \fP(S \times S')} \ I(\rho, \overline{\cM}_p) \ \leq \  \inf_{q \in \fP(S)} I(\sigma, \overline{\cN}_q) + \inf_{q' \in \fP(S')} I(\sigma', \overline{\cN}_{q'}). \label{prf:avqc_add_2}
     \end{align}
    We now drop the condition of finiteness on the sets $\fI$, $\fI'$, and notice that Caratheodory's Theorem allows to express each member of $\conv(\fI \otimes \fI')$ by a convex combination of finitely many 
    channels from $\fI \otimes \fI'$ (the same statement holds for $\conv(\fI)$, $\conv(\fI')$). Therefore, we conclude
    \begin{align}
     \inf_{\cM \in \conv(\fI \otimes \fI)} I(\rho, \cM) 
     &= \inf_{\substack{\tilde{S} \subset S \\ |\tilde{S}| < \infty}}\inf_{\substack{\tilde{S}' \subset S' \\ |\tilde{S}'| < \infty}} \inf_{p \in \fP(\tilde{S} \times \tilde{S}')} I(\rho, \tilde{\cM}_p) \label{prf:avqc_add_3}\\
     &\leq \inf_{\substack{\tilde{S} \subset S \\ |\tilde{S}| < \infty}}\inf_{\substack{\tilde{S}' \subset S' \\ |\tilde{S}'| < \infty}} \inf_{\substack{q \otimes q' \\ q \in \fP(\tilde{S}), q' \fP(\tilde{S}')}} 
     I(\rho, \tilde{\cM}_{q \otimes q'}) \label{prf:avqc_add_3_a}\\
     &\leq  \inf_{\substack{\tilde{S} \subset S \\ |\tilde{S}| < \infty}} \inf_{q \in \fP(\tilde{S})} I(\sigma, \overline{\cN}_{q}) + \inf_{\substack{\tilde{S}'  
         \subset S' \\ |\tilde{S}'| < \infty}} \inf_{q' \in \fP(\tilde{S}')} I(\sigma', \overline{\cN}_{q'}) \label{prf:avqc_add_4} \\
     &= \inf_{\tilde{\cN} \in \conv(\fI)} I(\sigma, \tilde{\cN}) + \inf_{\tilde{\cN}' \in \conv(\fI')} I(\sigma, \tilde{\cN}'). \label{prf:avqc_add_5}
    \end{align}
    The equalities in (\ref{prf:avqc_add_3}), (\ref{prf:avqc_add_5}) follow from Caratheodory's Theorem. The inequality in (\ref{prf:avqc_add_3_a}) arises from the fact that the minimization over all probability distributions
    on $S \times S'$ is replaced by minimization over the smaller set of product probability distributions. The inequality in (\ref{prf:avqc_add_4}) is by application of (\ref{prf:avqc_add_2}). Maximizing over all 
    $\rho \in \cS(\cH_A \otimes \cH'_A)$ together with application of Theorem \ref{theorem:avqc_ea_capacity} yields 
    \begin{align*}
     \overline{C}_{EA}^{AV}(\fI \otimes \fI') \leq \overline{C}_{EA}^{AV}(\fI) + \overline{C}_{EA}^{AV}(\fI'). 
    \end{align*}
    \end{proof}
    \begin{proof}[Proof of Corollary \ref{corr:avqc_cap_stab}]
   The claim follows immediately from Theorem \ref{theorem:avqc_ea_capacity}. Twofold application of Alicki-Fannes' inequality \cite{alicki04} implies for each state $\rho \in \cS(\cH_A)$ and channels $\cN, \cN' \in \cC(\cH_A,\cH_B)$,
   $\|\cN - \cN' \|_{\Diamond} := \gamma$
   \begin{align}
    \left|I(\rho, \cN) - I(\rho, \cN') \right| \leq 6 \gamma \log \dim \cH_A + h(\gamma).
   \end{align}
   We conclude 
   \begin{align}
    \left|C_{EA}^{AV}(\fI) - C_{EA}^{AV}(\fI')\right| 
    &\leq 6 D_{\Diamond}(\conv(\fI), \conv(\fI')) \dim \cH_A + h(D_{\Diamond}(\conv(\fI), \conv(\fI'))) \\
    &\leq 6 D_{\Diamond}(\fI, \fI') \dim \cH_A + h(D_{\Diamond}(\fI, \fI')).
   \end{align}
   \end{proof}
   We conclude this section by showing that the strong converse theorem \ref{theorem:strong_converse} can be easily inferred from strong converse statements for the corresponding capacities of perfectly known memoryless
   quantum channels. Strong converse bounds for this case have been shown in \cite{bennett14}. The following statement is the result from \cite{gupta14} rephrased to fit our notation. Actually, the assertions
   proven therein are even stronger than stated below. It has been shown that for each sequence of codes with rates strictly above the capacity, the transmission errors approach one with exponentially decreasing 
   trade-offs in the asymptotic limit. 
   \begin{theorem}[\cite{gupta14}, Theorem 11]
   Let $\cN \in \cC(\cH_A, \cH_B)$ be a channel. For each $\lambda \in (0, 1)$, and $R_e < \infty$, 
   \begin{align}
    \limsup_{n \rightarrow \infty} \ \frac{1}{n}\log \overline{N}_{EA}(n, \fI, R_e, \lambda) \ \leq \sup_{\rho \in \cS(\cH_A)} \ I(\rho, \cN).
   \end{align}
   \end{theorem}
   \begin{proof}[Proof of Theorem \ref{theorem:strong_converse}]
    It suffices to show the first claim of the theorem. Notice that we do not have to make a difference between $\conv(\fI)$ and its closure $\overline{\conv(\fI)}$. The capacity and error functions are continuous.  
    Since $I(\rho, \cN)$ is a convex-concave function, von Neumanns min-max Theorem \cite{vonneumann28} applies, and it holds
    \begin{align}
     \sup_{\rho \in \cS(\cH_A)} \ \inf_{\tilde{\cN} \in \overline{\conv(\fI)}} I(\rho, \tilde{N}) = \inf_{\tilde{\cN} \in \overline{\conv(\fI)}}  \sup_{\rho \in \cS(\cH_A)}  I(\rho, \tilde{N}).
    \end{align}
     We conclude 
     \begin{align}
      \limsup_{n \rightarrow \infty} \frac{1}{n} \log \overline{N}_{EA}^{AV}(n, \fI, R_e, \lambda)  
      &\leq  \inf_{\tilde{\cN} \in \overline{\conv(\fI)}} \limsup_{n \rightarrow \infty} \frac{1}{n} \log \overline{N}_{EA}(n, \tilde{\cN}, R_e, \lambda) \\
      &\leq  \inf_{\tilde{\cN} \in \overline{\conv(\fI)}} \sup_{\rho \in \cS(\cH_A)} I(\rho, \tilde{\cN})\\
      &=   \sup_{\rho \in \cS(\cH_A)} \inf_{\tilde{\cN} \in \overline{\conv(\fI)}} I(\rho, \tilde{\cN}).
     \end{align}
   \end{proof}
 
 \end{section}

\begin{section}{Conclusion}
In this work, we considered the task of entanglement-assisted classical message transmission over compound memoryless and arbitrarily varying quantum channels. For both channel models, we obtained single-letter capacity formulae. 
We have shown that the entanglement-assisted classical capacity is additive under composition of AVQCs and stable under perturbation of the generating set of channels. Both of these features fail to hold for the 
corresponding unassisted capacities in general. We demonstrated that the entanglement-assisted message transmission capacities obey no general converse bound for compound quantum channels. For arbitrarily varying quantum channels 
strong converse statements always hold. \newline 
An interesting question is, how these capacities do behave for the mentioned channel models, if the amount of shared entanglement provided for coding is limited. We leave the determination of a full trade-off relation between the 
optimal entanglement and message transmission rates. \newline 
It is known that optimal protocols for several other quantum communication tasks can be derived from coherent versions of entanglement-assisted message transmission codes in case of perfectly known memoryless quantum channels.
The codes developed for the compound quantum channel and AVQC models might be used to derive universal protocols for related quantum communication tasks also in case of these models of system uncertainty. 
\end{section}

\begin{appendix} \label{appendix}
\begin{section}{The encoding construction from \cite{hsieh08}} \label{appendix:encoding}
In this section, we give an account to the encoding maps introduced in \cite{hsieh08} which we use in Section \ref{subsect:prfs_compound} to derive codes sufficient for proving the achievability part of 
Theorem \ref{theorem:comp_ea_capacity}. First, we state the following well-known assertion about existence of certain Hilbert-Schmidt orthonormal bases of unitary matrices in $\cL(\bbmC^m)$, $m \in \bbmN$. 
\begin{lemma}\label{appendix:unitary_basis_lemma}
 There exists a family $\{v_j\}_{j=1}^{m^2} \subset \cL(\bbmC^m)$ of matrices with the properties 
 \begin{enumerate}
  \item $v_j^{\ast}v_j = \bbmeins_m$, for each $i \in [m]$, and
  \item $\tr(v_j^{\ast}v_l) = m \cdot \delta_{jl}$ for each $j,l \in [m]$.
 \end{enumerate}
\end{lemma}
Since the considerations in this appendix depend highly on the concept of frequency typical sets and subspaces, we state the corresponding definitions. For a finite alphabet $\cY$, and a probability distribution
$p \in \fP(\cY)$, the set of $p$-frequency typical sequences of length $k$ is defined by 
\begin{align}
 T_{p}^k := \left\{y^k \in \cY^k: \forall a \in \cY:\   \ k \cdot p(a) = N(a|y^k)\right\},
\end{align}
where $N(a|x^k)$ is the number of occurrences of the letter $a$ in $y^k$. A probability distribution $q \in \fP(\cY)$ is called a type of sequences in $\cY^k$, if $T_q^k \neq \emptyset$. If we denote the set of types 
in $\cY^k$ by $\fT_k$, it is well-known that
\begin{align}
 |\fT_k| \leq (k+1)^{|\cY|} 
\end{align}
holds. For further information on the concept of types, the reader is referred to \cite{csiszar11}. \newline 
Let $\sigma \in \cS(\cH)$, $\cH \simeq \bbmC^d$ be a density matrix with spectral decomposition
\begin{align*}
 \sigma := \sum_{i=1}^d \alpha_i \ \ket{\gamma_i}\bra{\gamma_i}
\end{align*}
counting zero eigenvalue eigenspaces. For each given $k \in \bbmN$, we can 
decompose $\cH^{\otimes k}$ into a direct product of frequency typical subspaces 
\begin{align}
 \cH^{\otimes k} = \bigoplus_{\lambda \in \fT_k} \cH_{\lambda}, \hspace{1cm} \cH_\lambda := \spann\left\{\gamma_{i^k}: i^k := (i_1,\dots,i_k) \in T_{\lambda}^k\right\}
\end{align}
for each $\lambda \in \fT_k$, where we have used the abbreviation $\fT_k := \cT(k,[d])$. We define shortcuts
\begin{align}
 d_\lambda := \dim \cH_\lambda = |T_\lambda^k|,\hspace{.4cm} \cX_\lambda := [d^2_\lambda] \times \{0,1\}, \hspace{.4cm} \text{and} \hspace{.4cm}\cX := \prod_{\lambda \in \fT_k} \cX_\lambda
\end{align}
Let, for each $\lambda \in \fT_k$, $\{v_{j_\lambda}^{\lambda} \}_{j_\lambda}^{d_\lambda^2} \subset \cL(\cH_\lambda)$ be a family of unitaries having the properties stated in 
Lemma \ref{appendix:unitary_basis_lemma} (applied with $m := d_\lambda$). 
We assume each matrix to be extended to the whole space $\cH^{\otimes k}$ by zero-padding. Define for each $\lambda \in \fT_k$ and $x_\lambda := (j_\lambda, r_\lambda)  \in \cX_\lambda$ a unitary 
$u_{x_\lambda}^{\lambda} := v_{j_\lambda}^\lambda \cdot (-1)^{r_\lambda}$. For each $x := (x_\lambda)_{\lambda \in \fT_k} \in \cX$, we define $u_x := \sum_{\lambda \in \fT_k} u_{x_\lambda}^{\lambda}$,
and a c.p.t.p. map $\tilde{\cE}_x \in \cC(\cH^{\otimes k},\cH^{\otimes k})$ by
\begin{align}
  \tilde{\cE}_x(a) :=  u_{x}a u_{x}^{\ast}  &&(a \in \cL(\cH^{\otimes k})).
\end{align}
The following nice properties of the family $\{\tilde{\cE}_x\}_{x \in \cX}$ we use throughout this paper, were shown in \cite{hsieh08}. If $\Psi := \ket{\psi}\bra{\psi}$ is the purification of $\sigma$ with 
state vector
\begin{align}
 \psi := \sum_{i=1}^d \sqrt{\alpha_i} \gamma_i \otimes \gamma_i \ \in \ \cH \otimes \cH,
\end{align}
it holds
\begin{align}
 \frac{1}{|\cX|} \sum_{x \in \cX} \  \tilde{\cE}_{x} \otimes \id_{\cH}^{\otimes k}(\Psi^{\otimes k}) = \sum_{\lambda \in \fT_k}   \alpha^k(T_\lambda^k) \ \pi_\lambda \otimes \pi_\lambda, \label{encode_1}
\end{align}
where $\alpha^k(T_\lambda^k) := \sum_{i^k \in T_\lambda^k} \ \alpha_{i_1}\cdot \dots \alpha_{i_k}$  and $\pi_\lambda := \frac{\bbmeins_{\cH_\lambda}}{d_\lambda}$ is the maximally mixed state on $\cH_\lambda$.
Moreover, the equality
\begin{align}
 S\left(\cN \circ \tilde{\cE}_x \otimes \id_{\cH}^{\otimes k}(\Psi^{\otimes k}) \right) = S\left(\cN \otimes \id_{\cH}^{\otimes k}(\Psi^{\otimes k}) \right)
\end{align}
holds for each $x \in \cX$, $\cN \in \cC(\cH^{\otimes k}, \cK)$ with $\cK$ being any (finite dimensional) Hilbert space was shown in \cite{hsieh08}.

\begin{proof}[Proof of Lemma \ref{lemma:mutual_approx}]
 We show that the family $\{\tilde{\cE}_x\}_{x \in \cX}$ introduced preceding this section has the desired properties. Note that by definition of $V$, it holds 
 \begin{align}
  \chi(q_\ast, V) = S(\overline{V}_{q_\ast}) - \sum_{x \in \cX} q_\ast(x)\ S\left(\cN \circ \tilde{\cE}_x \otimes \id_{\cH}^{\otimes k}(\Psi^{\otimes k}) \right),
 \end{align}
 with $\overline{V}_{q_\ast} := \sum_{x \in \cX} q_{\ast}(x) V(x)$. It holds
 \begin{align} 
  S\left(\overline{V}_{q_\ast}\right) 
  &=  S\left(\cN^{\otimes k} \otimes \id_{\cH}^{\otimes k}\left(\sum_{\lambda \in \cT_{k}} \alpha^k(T_\lambda^k) \ \pi_\lambda \otimes \pi_\lambda\right)\right) \label{approx_1}\\
  &\geq \sum_{\lambda \in \cT_{k}} \alpha^k(T_\lambda^k) \left( S(\cN^{\otimes k}(\pi_\lambda)) + S(\pi_\lambda) \right) \label{approx_2}\\
  &\geq S(\cN(\sigma)^{\otimes k}) + S(\sigma^{\otimes k}) - 2\cdot \log{|\fT_k|}\label{approx_3}.
 \end{align}
 The inequality in (\ref{approx_1}) is by definition of $V$ together with (\ref{encode_1}) and linearity of $\cN$, and the inequality in (\ref{approx_2}) is by concavity and additivity of $S$ for tensor product states. 
 The last of the above inequalities is by almost-convexity of $S$, i.e. the inequality
 \begin{align*}
  S(\overline{\tau}) \leq \sum_{x \in \cX} \ p(x) S(\tau_x) + H(p),
 \end{align*}
 which holds for each $p \in \fP(\cX)$ and set $\{\tau_x: x \in \cX \}$ of quantum states on a Hilbert space with average state $\overline{\tau} := \sum_{x \in \cX} \ p(x) \tau_x$. 
 Additionally, the fact that $\sum_{\lambda \in \fT_k} \alpha^k(T_\lambda^k) \pi_\lambda = \sigma^{\otimes k}$
 holds, was used. By similar reasoning as above, also the reverse inequality 
 \begin{align}
  S\left(\overline{V}_{q_\ast}\right) \leq S(\cN(\sigma)^{\otimes k}) + S(\sigma^{\otimes k}) + 2\cdot \log{|\fT_k|}
 \end{align}
 is proven. With the preceding bounds, and additivity of the quantum mutual information on inputs with tensor product structure, i.e. $I(\sigma^{\otimes k}, \cN^{\otimes k}) = k \cdot I(\rho, \cN)$, we have
 \begin{align}
  \left | k \cdot I(\sigma, \cN) - \chi(q_\ast, V)\right| \ = \ \left|I(\sigma^{\otimes k}, \cN^{\otimes k}) - \chi(q_\ast, V)\right| \ \leq \ 2\log|\fT_k| \ \leq \ 2d \cdot  \log (k+1)
 \end{align}
 where the rightmost inequality above is by type counting. 
\end{proof}
\end{section}

\end{appendix}


\begin{thebibliography}{10}
\bibitem{ahlswede06}
For classical AVCs, this was stated as an important open problem in information theory by R. Ahlswede in his  Shannon Award Lecture, ISIT '06, Seattle.
\end{thebibliography}
\end{document}